\documentclass[11pt]{article}
\usepackage[utf8]{inputenc}
\usepackage[margin= 1.2 in]{geometry}
\usepackage{lineno}
\usepackage{fontaxes}
\usepackage{trimspaces}
\usepackage{nccfoots}
\usepackage{setspace}
\usepackage{inconsolata}
\usepackage{libertine}
\usepackage{todonotes}

\usepackage[T1]{fontenc}
\usepackage{amsmath, amssymb}
\usepackage{xargs}
\usepackage{mathtools}
\usepackage{amsthm}
\usepackage{stmaryrd}
\usepackage{natbib}
\usepackage{graphicx}
\usepackage{hyperref}
\usepackage{xcolor}
\usepackage{comment}
\usepackage{todonotes}
\usepackage[algo2e,linesnumbered,lined,commentsnumbered,noend,boxed]{algorithm2e}

\SetCommentSty{FontInPseudocodeComments}
\usepackage{enumitem}
\usepackage{framed}
\usepackage[capitalise]{cleveref}
\usepackage[framemethod=TikZ]{mdframed}

\usepackage[noend]{algorithmic}
\usepackage[boxed]{algorithm}

\usepackage{mathtools}

\newcommand{\algcomment}[1]{\colorbox{black!10}{#1}}
\newcommand{\LineIf}[2]{ \STATE \algorithmicif\ {#1}\ \algorithmicthen\ {#2} }
\newcommand{\LineIfElse}[3]{ \STATE \algorithmicif\ {#1}\ \algorithmicthen\ {#2} \algorithmicelse\ {#3} }
\newcommand{\true}{ \ensuremath{\mathbf{true}}\xspace}
\newcommand{\false}{ \ensuremath{\mathbf{false}}\xspace}
\newcommand{\yes}{\ensuremath{\mathbf{yes}}\xspace}
\newcommand{\supp}{ \ensuremath{\mathrm{supp}} }

\newcommand{\dc}{\ensuremath{\mathord{\downarrow}}}
\newcommand{\uc}{\ensuremath{\mathord{\uparrow}}}

\newcommand{\defproblemu}[3]{
  \vspace{1mm}
\noindent\fbox{
  \begin{minipage}{0.963\columnwidth}
  #1 \\
  {\bf{Input:}} #2  \\
  {\bf{Question:}} #3
  \end{minipage}
  }
  \vspace{1mm}
}

\newtheorem{theorem}{Theorem}[section]
\newtheorem{definition}[theorem]{Definition}
\newtheorem{lemma}[theorem]{Lemma}

\newtheorem{claim}[theorem]{Claim}

\newtheorem{observation}[theorem]{Observation}

\newtheorem{hypothesis}[theorem]{Hypothesis}

\newcommand{\poly}{\mathtt{poly}}

\newcommand{\eps}{\varepsilon}

\newcommand{\cF}{\mathcal{F}}

\newcommand{\ptime}{\ensuremath{\mathsf{P}}\xspace}
\newcommand{\nptime}{\ensuremath{\mathsf{NP}}\xspace}

\title{An Invitation to ``Fine-grained Complexity of \nptime{}-Complete Problems''}
\author{Jesper Nederlof}

\begin{document}
\maketitle
\begin{abstract}
Assuming that \ptime is not equal to \nptime, the worst-case run time of any algorithm solving an \nptime-complete problem must be super-polynomial. But what is the fastest run time we can get? Before one can even hope to approach this question, a more provocative question presents itself: Since for many problems the na\"ive brute-force baseline algorithms are still the fastest ones, maybe their run times are already optimal?

The area that we call in this survey ``fine-grained complexity of \nptime-complete problems'' studies exactly this question. We invite the reader to catch up on selected classic results as well as delve into exciting recent developments in a riveting tour through the area passing by (among others) algebra, complexity theory, extremal and additive combinatorics, cryptography, and, of course, last but not least, algorithm design.
\end{abstract}
\section{Introduction}
\label{sec:intro}
A natural end goal of algorithm design is to obtain algorithms with \emph{optimal worst-case run times}.
More precisely, one aims for
\begin{enumerate}
\item algorithms that solve every instance $x$ of a fixed computational problem within time $T(s(x))$, where $T$ expresses the worst-case run time of the designed algorithm in terms of a chosen size measure $s(x)$ of the instance $x$, and
\item a proof that this worst-case run time constitutes a fundamental \emph{barrier}, i.e. no algorithm can achieve a run time $T'(s(x))$ for any function $T'$ that grows significantly smaller than $T$ (for example $T'(s(x))  = O(T(s(x))^{0.99})$).
\end{enumerate}
In modern terms, this topic can be described as \emph{fine-grained complexity}. However, in the current literature this term is mostly used for polynomial time algorithms. 

Assuming that \ptime is not equal to \nptime, for every \nptime-complete problem the optimal worst-case run time $T(s(x))$ is super-polynomial in the bit-size $|x|$. Following the Cobham–Edmonds thesis, researchers commonly aim to avoid such super-polynomial run times. The study of \emph{parameterized complexity} does this by introducing a size measure $s(x)$ that can be much smaller than $|x|$.
In this case, the central question of parameterized complexity is whether the super-exponential behavior can be isolated to depend only on $s(x)$ and whether one can design an algorithm with run time $T(s(x))|x|^{O(1)}$.

In this survey, we restrict our choice of $s(x)$ to canonical parameters of an input such as the number of vertices of a graph.
Since such size functions $s(x)$ are polynomially related to the number of bits $|x|$ by which $x$ is encoded, this renders many typical questions in parameterized complexity trivial.
But instead we study the fine-grained complexity of \nptime-complete problems and aim for the optimal worst-case run time $T(s(x))$, not discouraged by the prospect that probably $T$ will be super-polynomial (and in fact, probably even exponential!). 

It is tempting to assume that such optimal worst-case run times are established by very natural and extremely simple algorithms, based on empirical evidence and Occam's razor: For many canonical \nptime-complete problems, these simple algorithms still essentially have the fastest worst-case run time despite decades of research predating even the definition of \nptime-completeness. For example,
\begin{itemize}
	\item Karl Menger already asked in the 1930s~\cite{SCHRIJVER20051} whether the trivial $n!$ time algorithm for the \textsc{Traveling Salesperson Problem} that simply tries all round trips can be improved. This question was answered positively in the 1960s with a simple $2^{n}n^2$ time dynamic programming approach by Bellman~\cite{bellman1962dynamic}, Held and Karp~\cite{held1962dynamic}, but this is essentially still the fastest known worst-case run time for the problem.
	\item In the 1950s and 1960s several Russian scientists studied in a series of papers whether na\"ive baseline algorithms (under the Russian term ``perebor'', which translates to ``brute-force'' or ``exhaustive'') can be improved for $\nptime$-hard problems that include \textsc{$k$-CNF-Sat} (see Section~\ref{sec:sat} for a definition of~\textsc{$k$-CNF-Sat}). Yablonski, falling for the aforementioned temptation, even claimed (erroneously) to have a proof that brute-force methods cannot be avoided. See the survey~\cite{DBLP:journals/annals/Trakhtenbrot84} for details.
\end{itemize}
Being several decades of investigation wiser, researchers realized that the na\"ive assumption that simple baseline algorithms reach the barrier of optimal worst-case run times is far from the truth: For many classic \nptime-complete problems such as \textsc{Hamiltonian Cycle}, $k$-\textsc{Coloring}, \textsc{Bin Packing} and \textsc{Subset Sum} much exciting progress has been made that undermines (or in some cases, even disproves) the earlier belief that brute-force cannot be improved.
Simultaneously, for some computational problems, most notably \textsc{$k$-CNF-Sat}, the question is quickly getting more importance: The \emph{Strong Exponential Time Hypothesis} (SETH, see Hypothesis~\ref{eth}) that states roughly that brute-force is unavoidable for solving \textsc{$k$-CNF-Sat} is by now a well-accepted hypothesis and is often used as evidence that algorithms should have an optimal worst-case run time.

\paragraph{Exact Exponential Time Algorithms.}
While the above question has been studied for many \nptime-complete problems individually in the previous century, the area of studying the precise worst-case run time for \nptime-complete problems as a whole gained traction thanks to several influential surveys that featured an inspiring list of challenging open problems, authored by Woeginger in the beginning of the century~\cite{DBLP:conf/aussois/Woeginger01, DBLP:conf/iwpec/Woeginger04, DBLP:journals/dam/Woeginger08}. The years afterwards, the field flourished under the names ``(moderately/exact) exponential time algorithms'' and a series of Dagstuhl seminars devoted to the topic~\cite{husfeldt_et_al:DagRep.3.8.40,fomin_et_al:DagSemProc.08431.2,husfeldt_et_al:DagSemProc.10441.1} were held, the textbook ``Exact Exponential Algorithms''~\cite{DBLP:series/txtcs/FominK10} and two survey articles~\cite{DBLP:journals/cacm/FominK13, DBLP:journals/cacm/KoutisW16} were published in the journal Communication of the ACM. See also the survey~\cite{DBLP:conf/stacs/Schoning05} centered around \textsc{$k$-CNF-Sat}.

Since the name ``Exact Exponential Time'' algorithms doesn't really distinguish itself from parameterized complexity (i.e. parameterized algorithms for \nptime-complete problems are typically exact exponential time algorithms), and we believe that the main question studied also really connects to the younger field of ``fine-grained complexity'' (even though that seems to restrict itself to polynomial time algorithms), we use (yet) another term in this survey to refer to the subfield of theoretical computer science at hand: ``fine-grained complexity of \nptime-complete problems''.

\paragraph{A Warm-Up Algorithm for \textsc{$3$-Coloring}.} As an illustration of the type of questions one deals with in the area of fine-grained complexity of \nptime-complete problems we study the following:

\defproblemu{\textsc{$k$-Coloring}}{An undirected graph $G=(V,E)$.}{Is there a function $c:V \rightarrow [k]$ such that $c(v)\neq c(w)$ for every $\{v,w\} \in E$.}

Let us fix $k=3$ and try to design a fast algorithm for \textsc{$3$-Coloring}.
While the na\"ive baseline algorithm goes over all $3^n$ candidates for $c$ and hence runs in $O^*(3^n)$ time, a slightly smarter algorithm would iterate over $X \subseteq V$ and check whether $X$ is an independent set and whether $G[V \setminus X]$ is bipartite: it is easy to see that such an $X$ exists if and only if the sought function exists, and checking bipartiteness can be easily done in polynomial time. In fact, in this strategy we can even restrict the enumeration to sets $X$ of size at most $n/3$ to get an $O^*(\binom{n}{n/3})=O^*(1.89^n)$ time algorithm.

We now describe a smarter algorithm (originally suggested in~\cite{DBLP:journals/jal/BeigelE05}):
\begin{theorem}
There is a randomized algorithm that solves 
\textsc{$3$-Coloring} in $1.5^n$ time, and outputs a solution if it exists with probability at least $1-1/e$.
\end{theorem}
\begin{proof}
Consider a list-based variant of the $3$-\textsc{Coloring} problem in which we are given for every vertex $v \in V$ a list $L(v) \subseteq [3]$, and are looking for an assignment $c:V\rightarrow [k]$ such that $c(v) \in L(v)$ and $c(v)\neq c(w)$ for every $\{v,w\} \in E$. It can be easily shown that this problem can be solved in polynomial time if $|L(v)|\leq 2$ for each $v \in V(G)$, for example with a simple propagation algorithm or a reduction to $2$-\textsc{CNF-Sat} (which is also known to be solvable in polynomial time).

Now consider the following algorithm: For each vertex $v$, pick $L(v) \in \binom{[3]}{2}$ uniformly and independently at random and solve the resulting list-based variant in polynomial time. If it detects a function $c$, clearly it is correct. For the other direction, note that
\[
    \Pr_{L}[\forall_{v \in V} c(v) \in L(v) ] \geq (2/3)^n,
\]
since all lists are sampled independently.
Moreover, if $c(v) \in L(v)$ for all $v \in V$ then the list-based variant has a solution with the required properties. 
Hence, if we run $1.5^n$ trials of this polynomial time algorithm, the probability that we fail to output \yes if a solution $c$ exists is at most
\[
    (1-(2/3)^n)^{1.5^n} \leq 1/e,
\]
using the standard inequality $1+x \leq e^x$
\end{proof}

\paragraph{What this survey is (not) about.}
This survey aims to invite researchers in theoretical computer science into the field of fine-grained complexity of \nptime-complete problems.
We aim to convey that this is a beautiful field with elegant ideas and hosts many connections to other areas of theoretical computer science and mathematics.
To this end, we present proof sketches of a number of selected results. This includes both very recent works, to reflect the exciting and still developing character of the field, as well as older results that are too central and elegant to skip over. However, this survey by no means claims to be exhaustive. Some very important breakthroughs and research lines are omitted because, for example, there are already many other excellent surveys or textbooks discussing them. We will list a few of them in Section~\ref{sec:other}. An important emphasis of the survey, and the field in general, is its \emph{qualitative} character: Generally speaking, we deem results that improve run times beyond a natural barrier much more interesting than results that do not do this.

\section{Notation}
In the context of an instance $x$ of a computation problem or an input $x$ of an algorithm, we use $O^*()$ notation to omit factors that are polynomial in the length of the encoding of $x$, where we encode integers in binary.

If $b$ is a Boolean, we let $[b]$ denote the number $1$ if $b=\true$ and let it denote the number $0$ otherwise. On the other hand, if $i$ is an integer then we let $[i]$ denote the set $\{1,\ldots,i\}$.

For a set $S$ and integer $i$ we let $2^S$ denote the powerset of $S$ and let $\binom{S}{i}$ denote the family of subsets of $S$ of cardinality exactly $i$.
Using Stirling's approximation, it can be shown that
\begin{equation}\label{eq:binstir}
    \frac{n^n}{d^d (n-d)^{n-d}n^{O(1)}} \leq  \binom{n}{d} \leq \frac{n^n}{d^d (n-d)^{n-d}}n^{O(1)}.
\end{equation}
If $d=\alpha n$, then this is, up to factors polynomial in $n$, equal to
\[
    \frac{n^n}{d^d (n-d)^{n-d}} = \frac{n^n}{(\alpha n)^{\alpha n} ((1-\alpha)n)^{(1-\alpha)n}}=\left( \alpha^{-\alpha}  (1-\alpha)^{-(1-\alpha)}\right)^n = 2^{h(\alpha)n},
\]
where $h = - \alpha \lg \alpha - (1-\alpha)\lg(1-\alpha)$ is the binary entropy function. It can be shown with elementary calculus that
\begin{equation}\label{eq:entup}
    h(p) \leq p \lg (4/p).
\end{equation}

If $f$ is a function with $S$ as its domain and $X \subseteq S$ we denote $f(X) = \{f(x) : x \in X\}$. 

Matrices and vectors are denoted in boldface font. If $\textbf{a},\textbf{b}$ are two vectors of the same dimension $d$, we let $\langle a, b\rangle :=\sum_{i=1}^d \textbf{a}[i]\textbf{b}[i]$ denote their inner product.
If $\mathbf{M}$ is a matrix with rows indexed by $R$ and columns indexed by $C$, and $X\subseteq R$, $Y \subseteq C$ we let $\mathbf{M}[X,Y]$ be the submatrix of $\mathbf{M}$ formed by rows from $X$ and columns from $Y$. We also use $\mathbf{M}[X,\cdot]$ or $\mathbf{M}[\cdot,Y]$ to denote we do not restrict the rows/columns.

We let $i \equiv_p j$ denote that $i$ equals $j$ modulo $p$, omit $p$ if clear from context.

\section{Satisfiability of Conjunctive Normal Forms}\label{sec:sat}

One of the most well-studied \nptime-complete problems is that of determining the satisfiability of a boolean formula in Conjunctive Normal Form (CNF). Recall such formula is a conjunction of \emph{clauses}, which are disjunctions of \emph{literals}, where a literal is either a variable or its negation.
We say a boolean formula in CNF is a $k$-CNF if all clauses consist of at most $k$ literals.

\defproblemu{\textsc{$k$-CNF-Sat}}
{$k$-CNF formula $\varphi$ on $n$ variables and $m$ clauses.}
{Is there an assignment of the $n$ variables satisfying $\varphi$? }

The probably most famous hypotheses regarding the coarse/fine-grained complexity of \nptime-complete problems can now be formulated as follows:

\begin{hypothesis}[Exponential Time Hypothesis,\cite{DBLP:journals/jcss/ImpagliazzoP01}]\label{eth}
	There exists a $\delta >0$, such that no algorithm can solve \textsc{$3$-CNF-Sat} in $O(2^{\delta n})$ time.
\end{hypothesis}

\begin{hypothesis}[Strong Exponential Time Hypothesis,~\cite{DBLP:journals/jcss/ImpagliazzoP01}]\label{seth}
For every $\eps >0$, there exists a $k$ such that no algorithm can solve \textsc{$k$-CNF-Sat} in $O^*((2-\eps)^n)$ time.
\end{hypothesis}

While we use "fine-grained complexity" to refer to studying the possibility of an improvement of a $t(n)$ time bound to a $t(n)^{1-\Omega(1)}$ time bound, we can similarly use ``coarser-grained complexity'' to refer to studying the possibility of an improvement of a $t(n)$ time bound to a $t(n)^{o(1)}$ time bound.
It should be noted that, assuming Hypothesis~\ref{eth} we already have algorithms and lower bounds for many problems (parameterized by standard size measures) that are optimal in a coarser-grained manner. This includes all problems studied in this survey. For example, a $2^{o(n)}$ time algorithm for any problem in this survey (with $n$ being defined as in this survey as well) is known to refute the exponential time hypothesis. Such implications are typically a consequence of standard \nptime-completeness reductions and the sparsification lemma that we will discuss in detail below (Lemma~\ref{subsec:sparse}). See e.g.~\cite[Chapter 14]{DBLP:books/sp/CyganFKLMPPS15}.

\subsection{Algorithms for $k$-\textsc{CNF-Sat}}
\begin{algorithm}	
  \caption{Monien's and Speckenmeyer's algorithm for \textsc{$k$-CNF-Sat}.}
  \label{alg:ms}
	\begin{algorithmic}[1]
    \REQUIRE $\mathtt{MonienSpeckenmeyerkSAT}(\varphi)$ \hfill\algcomment{$\varphi$ is a $k$-CNF on $n$ variables}
    \ENSURE whether $\varphi$ is satisfiable
		\IF{there is an unsatisfied clause $l_1 \vee l_2 \vee \ldots \vee l_{k'}$}
            \FOR{$i=1,\ldots,k'$}\label{lin:loopi}
                \STATE $\rho \gets$ the restriction obtained by setting $\neg l_1,\neg l_2,\ldots, \neg l_{i-1}, l_i$\label{resdef}
                \LineIf{ $\mathtt{MonienSpeckenmeyerkSAT}(\varphi_{|\rho})$}{\algorithmicreturn \true}
            \ENDFOR
            \STATE \algorithmicreturn \false    
        \ENDIF
        \STATE \algorithmicreturn \true
	\end{algorithmic}
\end{algorithm}

The literature on the worst-case complexity of $k$-\textsc{CNF-Sat} is very rich.
There are several different algorithms that solve \textsc{$k$-CNF-Sat} in $O^*(2^{(1-1/O(k))n})$ time, but curiously it is not known whether this can be improved to a $O^*(2^{(1-1/o(k))n})$ time algorithm.
A lot of effort has been made to obtain small constants (for both constant $k$ and non-constant $k$) in the big-Oh term of the run time $O^*(2^{(1-1/O(k))n})$. We will not focus on such improvements here and refer to the state of the art~\cite{HansenKZZ19,Scheder24} for details.

We first describe some simple algorithms to solve \textsc{$k$-CNF-Sat}. The first one is slower than the second and the third (which are the state of the art, up to constants hidden in the big-Oh notation).

\subsubsection{Monien and Speckenmeyer's algorithm}
The algorithm by Monien and Speckenmeyer~\cite{MONIEN1985287} is the earliest one presenting a (modest) improvement over the trivial $O^*(2^n)$ time algorithm for \textsc{$k$-CNF-Sat}. It is outlined in Algorithm~\ref{alg:ms}. 
It uses the notion of a \emph{restriction}, which is a function $\rho: [n] \rightarrow \{0,1,*\}$ that sets a variable to $0,1$ or does not set it (corresponding to setting it to $*$).

The crucial step is in Line~\ref{resdef}, in which we define a restriction that all variables occurring in the first $i$ literals, and it does so in such a way that the first $i-1$ literals of a clause are not satisfied, but the $i$'th literal is satisfied.
Then it continues with determining whether the formula $\varphi_{|\rho}$ is satisfiable, which is the formula obtained by removing all clauses satisfied by $\rho$ and all variables set by $\rho$ (where the latter may result in an empty clause and hence an unsatisfiable formula).
If an assignment $\mathbf{x}$ satisfying $\varphi$ exists, then it will be detected at the iteration $i$ of the loop at Line~\ref{lin:loopi}, where $i$ is such that $\mathbf{x}$ satisfies $l_i$ but does not satisfy $l_1,\ldots,l_{i-1}$ 
If $T[n]$ is the number of recursive calls made by this algorithm, we have that $T[1]=1$ and $T[n] \leq \sum_{i=1}^k T[\max\{n-i,0\}]$. Hence, if we define $T'[1]=1$ and $T'[n] = \sum_{i=1}^k T'[\max\{n-i,0\}]$ then we have that $T[n]\leq T'[n]$ for all $n$. 

The numbers $T'[1],T'[2],\ldots$ are known as the \emph{Fibonacci $k$-step numbers}.
It can be shown with induction on $n$ or via combinatorial means\footnote{Using that $T'[n]$ is at most the number of subsets $X \subseteq [n]$ such that $X \cap \{i,\ldots,i+k\} \neq \emptyset$ for each $i \in [n-k]$, partition all of $[n]$ except at most $b$ elements into $n/b$ blocks of $b=\Theta(2^{k})$ consecutive elements $B_1,\ldots,B_{n/b}$ and argue that the number of options of $B_i\cap X$ is at most $2^{b}/c$ for some constant $c>1$.} that $T'[n]=2^{(1-2^{-O(k)})n}$.

\subsubsection{Sch\"oning's algorithm}
A considerably faster randomized\footnote{A derandomization based on covering codes is presented in~\cite{DBLP:journals/tcs/DantsinGHKKPRS02}.} algorithm by Sch\"oning is outlined in Algorithm~\ref{alg:ls}. A crucial ingredient is a subroutine $\mathtt{localSearch}(\varphi,\mathbf{x},d)$ that determines in $O^*(k^d)$ time whether $\varphi$ has a satisfying assignment of Hamming distance at most $d$ from $\mathbf{x}$. This subroutine is obtained with a small variant of Algorithm~\ref{alg:ms}, augmented with the following small observation: If there is an assignment $\mathbf{y}$ satisfying $\varphi$ but $\mathbf{x}$ does not satisfy $\varphi$ because some clause $l_1 \vee l_2 \vee \ldots l_p$ is not satisfied by $\mathbf{x}$, then $\mathbf{x}$ and $\mathbf{y}$ disagree in at least one of the $p$ variables of these literals. Hence we can recurse for $i=1,\ldots,p$ and flip the value of the variable underlying $l_i$ in $\mathbf{x}$ and assume that in one recursive call we moved to an assignment closer to $\mathbf{y}$ and hence decrease the distance parameter. Thus $\mathtt{localSearch}(\varphi,\mathbf{x},d)$ invokes at most $p \leq k$ direct recursive calls with distance parameter $d-1$, and hence the total number of (indirect) recursive calls invoked is at most $k^d$. Since any recursive call takes polynomial time, the claim $O^*(k^d)$ time follows.

Algorithm~\ref{alg:ls} now simply picks $\mathbf{x}$ at random and hopes that, if a solution $\mathbf{y}$ exists, that the Hamming distance $d(\mathbf{x},\mathbf{y})$ between $\mathbf{x}$ and $\mathbf{y}$ is at most $d$. This happens with probability at least $\sum_{i=0}^d\binom{n}{i}/2^n\geq \binom{n}{d}/2^n$. Hence after $2^n/\binom{n}{d}$ iterations the probability that $\true$ is returned is at least
\[
1- \left(1-\binom{n}{d}/2^n\right)^{\lceil 2^n/\binom{n}{d}\rceil } \geq 1-1/e \geq 1/2,
\]
where we use $1+x \leq e^x$ in the first inequality.
\begin{algorithm}	
  \caption{Sch\"oning's algorithm for \textsc{$k$-CNF-Sat}.}
  \label{alg:ls}
	\begin{algorithmic}[1]
    \REQUIRE $\mathtt{SchoningkSAT}(\varphi)$ \hfill\algcomment{$\varphi$ is a $k$-CNF on $n$ variables}
    \ENSURE \true with probability at least $\tfrac{1}{2}$ if $\varphi$ is satisfiable, and\false otherwise
		\STATE $d=n/(k+1)$ \label{lin:ksatadd}
		\FOR{$i=1 \ldots \lceil 2^n / \binom{n}{d} \rceil$}\label{lin:ksatloop}
			\STATE Pick $\mathbf{x} \in \{0,1\}^n$ uniformly at random \label{lin:ksatpick}
			\LineIf{$\mathtt{localSearch}(\varphi,\mathbf{x},d)$}{\algorithmicreturn \true} 
		\ENDFOR
		\STATE \algorithmicreturn \false
	\end{algorithmic}
\end{algorithm}

\noindent Now we use~\eqref{eq:binstir} to get that the run time $O^*(k^d 2^n/ \binom{n}{d})$ becomes
\[
\begin{aligned}
    O^*\left(k^d 2^n d^d (n-d)^{n-d}n^{-n}\right)&=O^*\left(2^n(kd)^d(n-d)^{n-d}n^{-n}\right)\\
    &=O^*\left(2^n \left( n\frac{k}{k+1} \right)^d \left(n\frac{k}{k+1} \right)^{n-d}n^{-n}\right)\\
    &=O^*\left(2^n \left(\frac{k}{k+1}\right)^n\right) = O^*\left(2^n \left(1-\frac{1}{k+1}\right)^n\right)\\
    &= O^*\left(2^{(1-1/O(k))n}\right),
\end{aligned}
\]
where we use $1+x \leq e^x$ in the last line.

\subsubsection{An algorithm based on random restrictions}
We give yet another algorithm with run time 
$O^*(2^{(1-1/O(k))n})$, based on random restriction and the switching lemma (often attributed to~\cite{DBLP:conf/stoc/Hastad86}, but the lemma builds on many previous works similar in spirit). An extension of the algorithm presented here that works for determining satisfiability of circuits of bounded depth can be found in~\cite{DBLP:journals/corr/abs-1107-3127}.\footnote{We are using a version presented in lecture notes by Valentine Kabanets \href{https://www2.cs.sfu.ca/~kabanets/407/lectures/lec11.pdf}{(link)}.}.

As before, let $\varphi$ be a $k$-CNF. Let the variables that occur in $\varphi$ be called $v_1,\ldots,v_n$, and let the clauses be called $C_1,\ldots,C_m$.

The switching lemma bounds the decision tree depth of random restrictions. For our application we need to be precise and use the following (somewhat lengthy) definition:

\begin{definition}
The \emph{canonical decision tree} of a CNF $\varphi$ is recursively defined by expanding a single root vertex $r$ into a tree as follows:
\begin{itemize}
\item if $\varphi$ has no clauses, the tree has $r$ as single vertex, referred to as a \emph{$1$-leaf},
\item if $\varphi$ has an empty clause the tree has $r$ as single vertex and is referred to as a \emph{$0$-leaf},
\item otherwise, define the \emph{level} of $r$ to be $0$ and its \emph{associated restriction} to be the restriction with $n$ stars. Let $i$ be the smallest integer such that $C_i$ is not yet satisfied and let its variables be $v'_1,\ldots,v'_p$. For $j=0,\ldots,p-1$, add for each vertex of level $j$ with associated restriction $\rho$ two children at level $j+1$ with as its associated restriction the restriction obtained from $\rho$ by setting $v'_i$ to either $\true$ or $\false$.\\
Recursively continue the construction for each vertex $r'$ at level $p$ with $\varphi$ being the associated restriction, attach the obtained canonical decision tree by identifying its root with $r'$.
\end{itemize}
\end{definition}

We let $CDTD(\varphi)$ denote the depth of the canonical decision tree of $\varphi$. It is not hard to see that this tree can be constructed in time linear in its size times factors polynomial in the input size, and since this is a binary tree its size is at most $2^{CDTD(\varphi)}$. Thus we can check in $O^*(2^{CDTD(\varphi)})$ time whether this tree has a $1$-leaf and hence determine whether $\varphi$ is satisfiable.

We now state the switching lemma. Many formulations circulate in the literature. We base ours on~\cite{beame1994switching}.

\begin{lemma}[Switching Lemma, \cite{DBLP:conf/stoc/Hastad86}]
\label{sl}
If $\varphi$ is a $k$-CNF formula and $\rho$ is a random restriction with $pn$ stars,\footnote{Both the version with exactly $pn$ stars and the version with independent star probability $p$ per variable are often stated in the literature. The current version is more handy for us.} then
\[
 \Pr_{\rho}[CDTD(\varphi_{|\rho}) \geq d] \leq (7kp)^d.
\]
\end{lemma}

\begin{algorithm}	
  \caption{Algorithm for \textsc{$k$-CNF-Sat} based on Switching Lemma (from~\cite{DBLP:journals/corr/abs-1107-3127}).}
  \label{alg:sl}
	\begin{algorithmic}[1]
    \REQUIRE $\mathtt{SwitchkSAT}(\varphi)$ \hfill\algcomment{$\varphi$ is a $k$-CNF on $n$ variables}
    \ENSURE whether $\varphi$ is satisfiable, and\false otherwise
		\STATE $p \gets 1/(30k)$
        \STATE Pick $S \in \binom{[n]}{pn}$ uniformly at random
		\FOR{each restriction $\rho$ that only assigns stars to $\{v_i: i \in S \}$}\label{linres}
			\STATE Construct the canonical decision tree of $\varphi_{|\rho}$\label{lindt}
			\LineIf{a $1$-leaf is encountered}{\algorithmicreturn \true}\label{oneleaf}
		\ENDFOR
		\STATE \algorithmicreturn \false
	\end{algorithmic}
\end{algorithm}

Using the above discussion and the equation $\mathbb{E}[X] = \sum_{i=0}^\infty \Pr[X \geq i]$ that holds for any integer non-negative random variable, we see that the expected run time of Lines~\ref{lindt} and~\ref{oneleaf} is
\[
\begin{aligned}
    \mathbb{E}_{\rho}[2^{|CDTD(\varphi_{|\rho})|}] &= \sum_{i=0}^\infty \Pr_\rho[2^{|CDTD(\varphi_{|\rho})|} \geq i]\\
    &= \sum_{i=0}^\infty \Pr_\rho[|CDTD(\varphi_{|\rho})| \geq \lg i]\\
    &\algcomment{Using Lemma~\ref{sl}}
    \\
    &\leq \sum_{i=0}^\infty (7 k p )^{\lg i}
    \leq  \sum_{i=0}^\infty \left(\tfrac{1}{4}\right)^{\lg i} = \sum_{i=0}^\infty 1/i^2 = O(1),
\end{aligned}
\]
where the last step is the standard convergence fact of $p$-series (and can be proved by approximating the series with an integral).
Now, since the number of restrictions considered on Line~\ref{linres} is at most $2^{n-|S|}$ we get by linearity of expectation that the expected run time of Algorithm~\ref{alg:sl} is $O^*\left(2^{(1-p)n}\right)=O^*\left(2^{(1-\frac{1}{O(k)})n}\right)$.

\subsection{Sparsification Lemma}\label{subsec:sparse}
One of the perhaps most impactful lemmas on the fine-grained hardness of \textsc{$k$-CNF-Sat} is the sparsification lemma. Loosely stated, it shows that for some function $f$, \textsc{$k$-CNF-Sat} in general is almost as hard as \textsc{$k$-CNF-Sat} for which the number of clauses is at most $f(k)n$, i.e. linear in the number of variables if we consider $k$ to be constant.

\paragraph{Statement.}
As in the original paper that presented the sparsification lemma, we formulate it in terms of hitting sets of set systems for convenience.
If $\cF \subseteq 2^U$ and $X \subseteq U$, say $X$ \emph{hits} $\cF$ if $X$ intersects every set in $\cF$.

\begin{lemma}[\cite{DBLP:conf/focs/ImpagliazzoPZ98, DBLP:conf/coco/CalabroIP06}]\label{lem:sparsify}
	There is an algorithm that, given $k \in \mathbb{N}$, $\varepsilon > 0$ and set family $\cF \subseteq 2^U$ of sets with size at most $k$, produces set systems $\cF_1,\ldots,\cF_\ell \subseteq 2^{U}$ with sets of size at most $k$ in $O^*(\ell)$ time such that
	\begin{enumerate}
		\item every subset $X \subseteq U$ hits $\cF$ if and only if $X$ hits $\cF_i$ for some $i$,
		\item for every $i=1,\ldots,\ell$ each element of $U$ is in at most $d=\left(\frac{4k^2 \lg(1/\varepsilon)}{\varepsilon}\right)^{k-1}$ sets of $\cF_i$,
		\item $\ell$ is at most $2^{2\varepsilon n}$.
	\end{enumerate}
\end{lemma}
In the original version of Lemma~\ref{lem:sparsify} the dependence on $k$ in item $2.$ was doubly-exponential; it was brought down to singly-exponential in~\cite{DBLP:conf/coco/CalabroIP06}.
It is an interesting question whether  this dependence can be further improved, even though some lower bounds in this direction are known~\cite{DBLP:conf/icalp/SanthanamS12}.

\paragraph{The case of graphs ($k=2$).}
A precursor of the sparsification lemma is a lemma from~\cite{DBLP:conf/soda/JohnsonS99} that states that \textsc{Vertex Cover} on sparse graphs is roughly as hard as \textsc{Vertex Cover} on general graphs (the paper actually speaks of the \textsc{Independent Set} problem, but this is equivalent to \textsc{Vertex Cover}, which is more naturally generalized to hypergraphs).

If $k=2$, then $\cF$ can be seen as a graph and $X$ hits $\cF$ exactly if it is a vertex cover of this graph, and we are in the aforementioned setting of~\cite{DBLP:conf/soda/JohnsonS99}. The proof idea in this setting is simple: If condition 2. of Lemma~\ref{lem:sparsify} does not hold, and we have a vertex of degree at least $d=\lg(1/\eps) / \eps$, then we branch on the decision whether we include this vertex $v$ in the vertex cover. We recurse on two subproblems, in one subproblem we include $v$ in the solution (lowering the number of vertices by $1$) and can remove it in the recursive call, whereas in the other subproblem we include all neighbors of $v$ in the solution and remove $v$ and all its neighbors from the recursive call (lowering the number of vertices by at least $d$). Doing this exhaustively, we generate at most $T[n]$ subproblems, where
\[
    T[n] \leq T[n-1] + T[n-d],
\]
and it can be shown that $T[n]\leq \binom{n}{n/d} \leq 2^{h(\eps/\lg(1/\eps))n}\leq 2^{2 \eps n}$, where we use~\eqref{eq:binstir} in the last step.

\paragraph{Relation with \textsc{$k$-CNF-SAT}.}
Lemma~\ref{lem:sparsify} is often stated in terms of CNF-formulas, but there is a simple reduction to the stated version. Since our version of the lemma is not directly about CNF-formulas we briefly illustrate one of the most useful consequences~\ref{lem:sparsify}:
We sketch how~\ref{lem:sparsify} can be used to show that a $2^{o(n+m)}$ time algorithm for \textsc{$k$-CNF-Sat} refutes Hypothesis~\ref{eth}. Suppose a $\delta>0$ with the condition of Hypothesis~\ref{eth} exists, and let $\varphi$ be a $3$-CNF on $n$ variables. Use Lemma~\ref{lem:sparsify} with $U$ to be all $2n$ literals (a positive and a negative literal for each variable) and set $\eps=\delta/5$, and $\mathcal{F}$ has a set for each clause of $\varphi$ consisting of its literals (being elements of $U$).

We obtain $2^{4 \delta n /5}$ set systems $\cF_1,\ldots,\cF_\ell$ such that every subset $X \subseteq U$ hits $\cF$ if and only if $X$ hits $\cF_i$ for some $i$.
We then cast every set family $\cF_i$ into a $k$-CNF $\varphi_i$ with one clause per set with literals being all elements of the set. 
Call a subset $X \subseteq U$ \emph{valid} if it includes exactly $1$ literal of $v_i$ and $\neg v_i$ for each $i$.
Applying condition 1. of Lemma~\ref{lem:sparsify} for all valid subsets, we have that $\varphi$ is satisfiable if and only if some $\varphi_i$ is satisfiable. Since the number of clauses in each $\varphi_i$ is $O(n)$, we can determine satisfiability of each $\varphi_i$ (and hence of $\varphi)$ in total time $2^{4\delta n/5}2^{o(n+m)}\leq O(2^{5\delta/6})$ time with the assumed subexponential time algorithm, contradicting Hypothesis~\ref{eth}.

While this addresses coarser-grained reductions (conditioned on 
Hypothesis~\ref{eth}), Lemma~\ref{lem:sparsify} is also very often used as first step for more fine-grained reduction (conditioned on Hypothesis~\ref{seth}) such as e.g. the reductions from~\cite{DBLP:journals/talg/AbboudBHS22,DBLP:journals/talg/CyganDLMNOPSW16} .

\subsubsection{The Algorithm}
Before we describe the algorithm, we need the following definition: 
\begin{definition}[Flowers and Petals]\label{def:flower} For brevity, we refer to a set of size $s$ an \emph{$s$-set}.
	An \emph{$s$-flower} is a collection of $s$-sets $S_1,\ldots,S_z$ such that the \emph{heart} $H:=\cap_{i=1}^z S_i$ is non-empty. The sets $S_1 \setminus H,\ldots,S_z \setminus H$ are referred to as the \emph{petals}; the quantity $|S_i \setminus H|$ (which is independent of $i$ in an $s$-flower) is called the \emph{petal size}. 
\end{definition}

The algorithm from Lemma~\ref{lem:sparsify} is in Algorithm~\ref{alg:reduce}. The constants $\theta_1,\ldots,\theta_k$ will be defined later. A flower $S_1,\ldots,S_z$ with petal size $p$ is called \emph{good} if $z \geq \theta_p$. We let $\pi(\cF)$ denote the set of inclusion-wise minimal sets of $\cF$. 

\begin{algorithm}[H]
	\caption{Algorithm implementing Lemma~\ref{lem:sparsify}.}
	\label{alg:reduce}
	\begin{algorithmic}[1]
		\REQUIRE $\mathtt{reduce}(\cF)$ \hfill\algcomment{Assumes the sets of $\cF$ do not contain each other ($\cF=\pi(\cF)$)} 
		\ENSURE A collection of set systems as promised in Lemma~\ref{lem:sparsify}.
		\FOR{$s=2,\ldots,k$}
			\FOR{$p=1,\ldots,s-1$}
				\item[] \algcomment{check if there exists a good $s$-flower, and if so branch on it} 
				\IF{there exists an $s$-flower $S_1,\ldots,S_z$ with petal-size $p$ and $z \geq \theta_p$}
					\STATE $H \gets \cap_{i=1}^z S_i$
                    \STATE $\cF_{heart} \gets \pi(\cF \cup \{H \})$
                    \STATE $\cF_{petals} \gets \pi(\cF \cup \{S_i \setminus H: i=1,\ldots,z\})$.
                    \item[] \algcomment{branch on whether we hit the heart or \emph{all} of the petals}
					\STATE \algorithmicreturn\ $\mathtt{reduce}( \cF_{heart} ) \cup \mathtt{reduce}(\cF_{petals})$. \label{sparsify:branch}
				\ENDIF
			\ENDFOR
		\ENDFOR
		\STATE \algorithmicreturn\ $\{\cF\}$.
	\end{algorithmic}
\end{algorithm}
Note that $\pi$ preserves the family of hitting sets and hence every $X \subseteq U$ hits $\cF$ if and only if it hits $\cF_{heart}$ (if $H \cap X \neq \emptyset$) or $\cF_{petals}$ (if $H \cap X = \emptyset$), and thus Item $1$ of Lemma~\ref{lem:sparsify} is indeed true.
\begin{observation}\label{obs:jflower}
	If $\cF$ has no good $j$-flower, each $h$-set is contained in at most $\theta_{j-h}-1$ sets of size $j$.
\end{observation}
To see that this is true, note that if some $h$-set $H$ is contained in at least $\theta_{j-h}$ sets of size $j$, then these $\theta_{j-h}$ sets will form a $j$-flower with $H$ as heart and since the petal size is $j-h$, it will be good. 

Using $h=1$, we see that every element of $U$ is in at most $\theta_{k-1}$ sets of $\cF$ if $\cF$ is output by $\mathtt{reduce}$ as in this case it does not have good $j$-flowers for every $j \leq k$. Hence, it follows from Observation~\ref{obs:jflower} that Item $2$ of Lemma~\ref{lem:sparsify} is satisfied as long as $\theta_{k-1}\leq d$.

\subsubsection{Bounding the run time and output size.}
The original proof (and also the proof in the textbook~\cite{DBLP:series/txtcs/FlumG06}) features several claims about clauses being subsequently added and removed several times along recursive paths of Algorithms~\ref{alg:reduce}. We employ an (arguably) more direct approach.\footnote{After sending this survey for review, the author learned that another write-up of a more direct proof of the sparsification lemma appeared at the conference SOSA'26~\cite{sparsificationSOSA}.}

We bound the run time of $\mathtt{reduce}(\cF)$ by assigning a potential function $\varphi$ to a set system $\cF$ which indicates the progress towards sparsification. The potential function $\varphi$ is defined as follows:
\begin{definition} Let $\alpha$ be some parameter to be set later, $\beta_j=(4\alpha k)^{j-1}$ and $\theta_j=\alpha\beta_j$. Define $\sigma(\cF)$ to be the largest $s <k$ such that every $h$-set is in at most $2\theta_{j-h}$ sets of size $j$ for every $j \leq s$. Define $\varphi(\cF)=\sum_{j=1}^{\sigma(\cF)} \beta_{k-j+1}|\cF_j|$, where $\cF_j$ denotes all $j$-sets of $\cF$.
\end{definition}

Clearly, $\varphi(\cF)$ is non-negative. We proceed by showing that in every recursive call $\varphi$ increases substantially and upper bounding $\varphi$ if $\mathtt{reduce}$ terminates. For the upper bounding part, note that
\begin{equation}\label{eq:upppot}
\begin{aligned}
\varphi(\cF) &= \sum_{j=1}^{\sigma(\cF)} \beta_{k-j+1}|\cF_j| \\
&\leq \sum_{j=1}^{\sigma(\cF)} \beta_{k-j+1} (n 2\theta_{j-1})&\hfill\algcomment{each element is in at most $2\theta_{j-1}$ sets of size $j$} \\
&\leq k\beta_{k-j+1}\alpha2\beta_{j-1}n & \hfill\algcomment{$\theta_j=\alpha\beta_j$} \\
&\leq k\alpha2\beta_{k-1}n & \hfill\algcomment{$\beta_j=(4\alpha k)^{j-1}$} \\
&\leq \beta_k n.
\end{aligned}
\end{equation}

Now we will consider a particular recursive call of $\mathtt{reduce}$ and lower bound the increase of the potential. As $\varphi(\cF)$ involves a sum up to $\sigma(\cF)$, it will be useful to first show that $\sigma(\cF)$ never decreases. Intuitively, $\sigma$ can be thought of as the \emph{phase} of the algorithm: the sets of size at most $\sigma(\cF)$ are processed already by the algorithm and are guaranteed to be nice in the sense that they cannot be extremely concentrated.

\begin{lemma}\label{lem:sigma}
	$\sigma(\cF) \leq \min \{\sigma(\cF_{heart}),\sigma(\cF_{petals})\}$.
\end{lemma}
\begin{proof}
	For obtaining a contradiction, suppose that $j=\min \{\sigma(\cF_{heart}),\sigma(\cF_{petals})\}+1 \leq \sigma(\cF)$. By the definition of $\sigma$ this means that a set of size $j$ must have been added, i.e. the picked $s$-flower either has heart-size $j$ or petal-size $j$ so that in either $\cF_{heart}$ or $\cF_{petals}$ some $h$-set is in at least $2\theta_{j-h}$ sets of size $j$. As $j < s$, $\cF$ does not contain a good $j$-flower and thus every $h$-set is contained in at most $\theta_{j-h}-1$ sets of size $j$ in $\cF$. Thus in either $\cF_{heart}$ or $\cF_{petals}$, more than $\theta_{j-h}$ supersets of size $j$ of a fixed $h$-set $X$ are added. For $\cF_{heart}$ only $H$ was added so this is not possible. For $\cF_{petals}$ this means $X$ is contained in at least $l=\theta_{j-h}$ of the petals $S_1\setminus H,\ldots,S_z\setminus H$ that are without loss of generality $S_1\setminus H,\ldots,S_l\setminus H$. But then the set $S_1,\ldots, S_l$ also is an $s$-flower, and it has heart $H\cup X$. This gives a contradiction with $S_1,\ldots,S_z$ being picked as $S_1,\ldots,S_l$ has petal size $j-h$ so it is good since $l= \theta_{j-h}$, and since it has smaller petal-size it would have been preferred by the algorithm.
\end{proof}
By Lemma~\ref{lem:sigma} we do not need to worry about $\sigma$ when analyzing the increase of $\varphi$, as $\sigma$ can only increase (which only increases $\varphi$). We proceed by lower bounding the increase of $\varphi(\cF_{heart})$ compared to $\varphi(\cF)$. Note that in $\cF_{heart}$ an $h$-set is added, and at most $2\theta_{j-h}$ sets of size $j\leq \sigma(\cF)$ are removed due to the $\pi$ operation and the addition of the heart $H$ as $H$ is in at most $2\theta_{j-h}$ sets of size $j\leq \sigma(\cF)$ (note that we do not have to account for $j$-sets with $j > \sigma(\cF)$ as they do not contribute to $\varphi(\cF)$). Therefore we have that
\begin{equation}
\begin{aligned}\label{eq:heartpot}
	\varphi(\cF_{heart}) &\geq \varphi(\cF) + \beta_{k-h+1} - \sum_{j=h+1}^{\sigma(\cF)} 2\theta_{j-h}\beta_{k-j+1}\\
 					  &\geq \varphi(\cF) + \beta_{k-h+1} - 2\alpha k \beta_{k-h} \hfill\ \ \ \ \ \algcomment{$\theta_j =\alpha\beta_j$ and $\beta_j = (4\alpha k)^{j-1}$}\\
      							  &> \varphi(\cF).
\end{aligned}
\end{equation}
Similarly comparing $\cF_{petals}$ with $\cF$, note that the $z$ petals of size $p$ are added and each such set is contained in at most $2\theta_{j-p}$ sets of size $j$. Hence, we obtain
\begin{equation}
\begin{aligned}\label{eq:petalpot}
	\varphi(\cF_{petals}) &\geq \varphi(\cF) + z\left(\beta_{k-p+1}-\sum_{j=p+1}^{\sigma(\cF)}2\theta_{j-p}\beta_{k-j+1}\right)\\
 &\ \ \ \algcomment{$\theta_j =\alpha\beta_j$ and $\beta_j = (4\alpha k)^{j-1}$ }\\
							   &\geq \varphi(\cF) + z(\beta_{k-p+1} - 2\alpha k \beta_{k-p})\\
          &\ \ \  \algcomment{$S_1,\ldots,S_z$ is a good $s$-flower with petal size $p$}\\
							   &\geq \varphi(\cF) + \theta_p\beta_{k-p+1}/2\\
  							   &\geq \varphi(\cF) + \alpha\beta_p\beta_{k-p+1}/2\\
   							   &\geq \varphi(\cF) + \tfrac{1}{2}\alpha\beta_k.
\end{aligned}
\end{equation}

Define $T(m)$ as the number of recursive calls of $\mathtt{reduce}(\cF)$ if $m=\varphi(\cF)$.
By~\eqref{eq:upppot} we have that $T(m)=0$ if $m \geq \beta_k n$ and otherwise by~\eqref{eq:heartpot} and~\eqref{eq:petalpot} we have that
\[
	T(m) \leq T(m+1) + T(m+\tfrac{1}{2}\alpha\beta_k).
\] 
Since $\beta_k n/(\tfrac{1}{2}\alpha\beta_k)=2n/\alpha$, this implies that $T(0)\leq \sum_{i=0}^{2n/\alpha}\binom{\beta_k n}{i} \leq \frac{2n}{\alpha}\binom{\beta_k n}{2n/\alpha}$.
Writing this in terms of the binary entropy function $h(p)=p\lg \tfrac{1}{p}+(1-p)\log\tfrac{1}{1-p}$ by using the inequality $\binom{n}{k} \leq 2^{h(k/n)n}$ we obtain that $T(0)\leq \frac{2 n}{\alpha}2^{\lambda n}$ for
\begin{equation}
\begin{aligned}\label{eq:petalpot2}
	\lambda	&\leq h\left(\frac{2}{\alpha\beta_k}\right)\beta_k & \\
			&\leq \frac{2}{\alpha}\lg(2\alpha\beta_k)  & \hfill\algcomment{Using $h(p)\leq p\lg(4/p)$~\eqref{eq:binstir}}\\
			&\leq \frac{4(k-1)\lg(8\alpha k)}{\alpha}  & \hfill\algcomment{Using $\beta_{k}=(4\alpha k)^{k-1}$}\\
\end{aligned}
\end{equation}
Using $\alpha \approx k\lg(1/\varepsilon)/\varepsilon$ we have that $\lambda\leq  2\varepsilon$ (for small enough $\varepsilon$) and we may pick $d$ to be
\[
 \theta_{k-1} = \alpha(4\alpha k^2)^{k-2} \leq (4\alpha k^2)^{k-1} = \left(\frac{4k^2 \lg(1/\varepsilon)}{\varepsilon}\right)^{k-1}.
\]
This finishes the proof of Lemma~\ref{lem:sparsify}.

\section{\textsc{Set Cover} and its Special Cases}\label{sec:setcov}
Much recent progress has been made on variants of the following well-known \nptime-complete problems:

\defproblemu{\textsc{Set Cover}}
{A universe $U$, a set system $\mathcal{S} \subseteq 2^U$ with $|U|=n$, and an integer $k$.}
{Are there sets $S_1,\ldots,S_k \in \mathcal{S}$ such that $\cup_{i=1}^ k S_i = U$?}

\defproblemu{\textsc{Set Partition}}
{A universe $U$, a set system $\mathcal{S} \subseteq 2^U$ with $|U|=n$, and an integer $k$.}
{Are there \underline{disjoint} sets $S_1,\ldots,S_k \in \mathcal{S}$ such that $\cup_{i=1}^ k S_i = U$?}

\subsection{\textsc{Set Cover} in $2^{n}n^{O(1)}$ time with Yates's algorithm and Inclusion/Exclusion}
We will use boldface font to indicate that a variable is a matrix or a vector. For an integer $s$, we let $\mathbf{I}_s$ denote the identity matrix with dimensions equal to $s$.
Fix a field $\mathbb{F}$. Given matrices $\mathbf{A}\in \mathbb{F}^{R\times C}, \mathbf{B} \in \mathbb{F}^{R_2\times C_2}$, we define the \emph{Kronecker product} $\mathbf{A} \otimes \mathbf{B}$ as the matrix whose rows and columns are indexed with $R \times R_2$ and $C \times C_2$ with entry $(\mathbf{A} \otimes \mathbf{B})[(r,r_2),(c,c_2)] = \mathbf{A}[r,c]\cdot\mathbf{B}[r_2,c_2]$.
The fact that the Kronecker product and normal matrix product distribute is often called the \emph{mixed product property}:
\[
    (\mathbf{A} \otimes \mathbf{B})(\mathbf{C} \otimes \mathbf{D})=\mathbf{A}\mathbf{C}\otimes\mathbf{B}\mathbf{D}.
\]
The \emph{$n$th Kronecker power $\mathbf{A}^{\otimes n}$} is the product of $n$ copies of $\mathbf{A}$, so the matrix that has its rows indexed by $R^n$ and its columns indexed by $C^n$ and entries defined by
\[
	\mathbf{A}^{\otimes n}[r_1,\ldots,r_n,c_1,\ldots,c_n]=\prod_{i=1}^n \mathbf{A}[r_i,c_i].
\]
Kronecker powers will be useful for us by virtue of the following lemma:
\begin{lemma}[Yates' Algorithm~\cite{yates1937design}]\label{lem:yates}
	Let $\mathbf{A} \in \mathbb{F}^{ R \times C}$, $n$ be an integer and $\mathbf{v} \in \mathbb{F}^{C^n}$ given as input. Then $\mathbf{A}^{\otimes n}\mathbf{v}$ can be computed in $O(\max\{|R|^{n+2},|C|^{n+2}\})$ time.
\end{lemma}
\begin{proof}
Let $r=|R|$ and $c=|C|$.
The lemma is proved by a simple Fast Fourier Transform style procedure. We follow the presentation of \cite[Section 3.1]{Kaski18}. Observe that by the mixed product property it follows that $\mathbf{A}^{\otimes n} = \mathbf{A}^{[n-1]}\mathbf{A}^{[n-2]}\cdots\mathbf{A}^{[0]}$, where
\begin{equation}
\label{yates}
    \mathbf{A}^{[\ell]} = \underbrace{\mathbf{I}_r \otimes \mathbf{I}_r \otimes \ldots \otimes \mathbf{I}_r}_{n-\ell-1}\otimes\; \mathbf{A} \otimes \underbrace{\mathbf{I}_c \otimes \mathbf{I}_c \otimes \ldots \otimes \mathbf{I}_c}_{\ell}.
\end{equation}
Now the algorithm that computes $\mathbf{A}^{\otimes n}\mathbf{v}$ na\"ively evaluates the following expression:
\[
	\mathbf{A}^{[n-1]}\left(\mathbf{A}^{[n-1]}\left(\ldots\mathbf{A}^{[1]}\left(\mathbf{A}^{[0]}\mathbf{v}\right)\right)\right).
\]
The runtime bound follows from the observation that the number of arithmetic operations needed to multiply $\mathbf{A}^{[\ell]}$ with a vector is at most $O(r^{n-\ell}c^{\ell+1})$.
\end{proof}
We consider vectors indexed by the power set $2^U$ (in lexicographical order), and the matrices
\[
\mathbf{Z} = \left(\begin{tabular}{cc} 
1&0\\
1&1
\end{tabular}\right) \qquad \text{and \qquad}
\mathbf{M} = \left(\begin{tabular}{cc} 
1&0\\
-1&1
\end{tabular}\right)
\]
will play a crucial role.
In particular, if $\mathbf{v}$ is indexed by $2^{U}$ then $\mathbf{Z}^{\otimes n}\mathbf{v}$ is the so-called \emph{zeta} transform of $\mathbf{v}$: it satisfies $(\mathbf{Z}^{\otimes n}\mathbf{v})[Y]=\sum_{X \subseteq Y}\mathbf{v}[X]$ for any $Y$. Similarly, the \emph{M\"obius} transform $(\mathbf{M}^{\otimes n}\mathbf{v})[Y]=\sum_{X \subseteq Y}(-1)^ {|Y\setminus X|}\mathbf{v}[X]$ for any $Y$. Since $\mathbf{M}\mathbf{Z}=\mathbf{I}_2$ we have by the mixed product property that $\mathbf{M}^{\otimes k}$ and $\mathbf{Z}^{\otimes k}$ are inverses of each other for each $k$.\footnote{The zeta and M\"obius transformations typically are more generally defined in the context of an arbitrary partial order set; we restrict our attention to these transformations in the special case of the subset lattice.}

\begin{theorem}[\cite{DBLP:journals/siamcomp/BjorklundHK09}]\label{thmsc}
\textsc{Set Cover} can be solved in time $2^nn^{O(1)}$.
\end{theorem}
\begin{proof}
Let $\mathbf{v}$ be the indicator vector of $\mathcal{S}$ (i.e. $\mathbf{v}[X]=1$ if $X \in \mathcal{S}$ and $\mathbf{v}[X]=0$ otherwise).
\begin{enumerate}
    \item Compute $\mathbf{z}=\mathbf{Z}^{\otimes n}\mathbf{v}$
    \item Compute the vector $\mathbf{z'}$ defined as $\mathbf{z'}[X]=(\mathbf{z}[X])^k$
    \item Compute $\mathbf{M}^{\otimes n}\mathbf{z'}$, and output\true if and only if $(\mathbf{M}^{\otimes n}\mathbf{z'})[U]>0$.
\end{enumerate}
Using Lemma~\ref{lem:yates} for Step 1. and Step 3. leads to a $2^{n}n^{O(1)}$ time algorithm.
To see correctness, let us define another vector $\mathbf{s}$ indexed by $2^U$ as follows. For every $X \subseteq U$ we let
\[
 \mathbf{s}[X]= \left|\left\{ (S_1,\ldots,S_k) \in \mathcal{S}^k : \bigcup_{i=1}^ k S_i = X \right\}\right|,
\]
and note that the instance of \textsc{Set Cover} is a yes-instance precisely when $\mathbf{s}[U]>0$.
Observe that for any $X\subseteq U$ we have that
\[
\begin{aligned}
	\mathbf{z'}[X] &=\left|\left\{ (S_1,\ldots,S_k) \in \mathcal{S}^k : S_i \subseteq X \text{ for all } i \in\{1,\ldots,k\}\right\}\right| \\
	&= \left|\left\{ (S_1,\ldots,S_k) \in \mathcal{S}^k : \bigcup_{i=1}^ k S_i \subseteq X \right\}\right|  = (\mathbf{Z}^{\otimes n}\mathbf{s})[X].
\end{aligned}
\]
Hence, $(\mathbf{M}^{\otimes n}\mathbf{z'})[U]=(\mathbf{M}^{\otimes n}\mathbf{Z}^{\otimes n}\mathbf{z'})[U]=\mathbf{s}[U]$, and the correctness follows.
\end{proof}

\subsection{Avoiding Computation over the Whole Subset Lattice}
For a set family $\cF$ we denote $\dc\cF:= \{ X : X \subseteq F, F \in \cF\}$ and $\uc\cF:= \{ X: F \subseteq X, F \in \cF\}$. Often $\dc\cF$ is called the \emph{down-closure} of $\cF$ and $\uc\cF$ is called the \emph{up-closure} of $\cF$.
For a vector $\mathbf{v}$, we let $\supp(\mathbf{v})$ denote the set of indices of $\mathbf{v}$ at which $\mathbf{v}$ has a non-zero value.
Revisiting the proof of Theorem~\ref{thmsc} and in particular Lemma~\ref{lem:yates} it can be observed that, since both $\mathbf{Z}$ and $\mathbf{M}$ are lower-triangular, $\uc\supp(\mathbf{A}^{[\ell]}\mathbf{v}) \subseteq \uc\supp(\mathbf{v})$ in~\eqref{yates} for the special case that $\mathbf{A}=\mathbf{Z}$ or $\mathbf{A}=\mathbf{M}$. Restricting the algorithm from the proof of Lemma~\ref{lem:yates} to only compute entries indexed by $\uc\supp(\mathbf{v})$ therefore directly leads to an algorithm that computes $\mathbf{M}^{\otimes n}\mathbf{v}$ and $\mathbf{Z}^{\otimes n}\mathbf{v}$ in time $|\uc\supp(\mathbf{v})| n^{O(1)}$.
\begin{lemma}[\cite{DBLP:journals/corr/abs-0802-2834}]
\label{lem:trimmob}
\textsc{Set Cover} can be solved in time $|\uc\mathcal{S}|n^{O(1)}$.
\end{lemma}
A well-studied special case of \textsc{Set Cover} (though, with exponentially many sets) is the \textsc{$k$-Coloring} problem as defined in Section~\ref{sec:intro}.  
The \textsc{$k$-Coloring} problem can be thought of an instance of $\textsc{Set Cover}$ where $\mathcal{S}$ is the family of all (inclusion-wise maximal) independent sets of $G$. 
As such, Theorem~\ref{thmsc} gave the first $O^*(2^n)$ time algorithm for \textsc{$k$-Coloring}, presenting a new natural barrier of worst-case complexity after a long line of research.\footnote{Actually, unlike for $k$-\textsc{CNF-Sat} and \text{Set Cover} with sets of size at most $k$, we currently do not even know how to solve \textsc{$k$-Coloring} for $k=7$ in time $O((2-\eps)^n)$ for some constant $\eps>0$. For $k=5,6$, such improved algorithms were only found recently (see~\cite{Zamir21}).}

In the case where $\mathcal{S}$ is defined as the set of maximal independent sets, we have that all sets in $\uc\mathcal{S}$ are dominating sets of $G$. Using Shearer's entropy lemma, it was shown in~\cite{DBLP:journals/corr/abs-0802-2834} that graphs of maximum degree $d$ have at most $(2^{d+1}-1)^{n/(d+1)}$ dominating sets. Therefore Lemma~\ref{lem:trimmob} gives the following result:

\begin{theorem}[\cite{DBLP:journals/corr/abs-0802-2834}]\label{colbd}
For every $d$, \textsc{$k$-Coloring} on graphs of maximum degree $d$ can be solved in time $O^*((2^{d+1}-1)^{n/(d+1)})$.
\end{theorem}

We will showcase more applications of the idea behind Lemma~\ref{lem:trimmob}, and use the following extension (which is a slight extension of observations made in~\cite{DBLP:journals/corr/abs-0802-2834}):

\begin{lemma}[Crossing the Middle layer]
\label{lem:crossmid}
Let $\mathcal{S}$ be downward closed (i.e. if $S \in \mathcal{S}$ and $S' \subseteq S$, then $S' \in \mathcal{S}$) and let $\mathtt{A_o}$ be an oracle algorithm that determines in $\mathtt{T_o}$ time whether a given set $S \subseteq U$ belong to $\mathcal{S}$.
Then there is an algorithm that takes as input a set family $\mathcal{F} \subseteq \binom{U}{\geq |U|/2}$, and determines whether there exist $S_1,\ldots,S_k \in \mathcal{S}$, $i \in [k]$ and $F \in\mathcal{F}$ such that $\cup_{j=1}^{i-1}S_j \supseteq F$ and $\cup_{j=i}^{k}S_j \supseteq U \setminus F$.

Moreover, for every $\eps >0$ there exists an $\eps'>0$ such that, if $|\dc\mathcal{F}| \leq (2-\eps)^n$, then the algorithm runs in time $O((2-\eps')^n\cdot \mathtt{T_o})$ time.
\end{lemma}
\begin{proof}[Proof sketch]
By restricting the algorithm in the proof of Lemma~\ref{lem:yates} to only compute table entries indexed by $\dc\mathcal{F}$ we get the following result:
\begin{claim}\label{clmmid}
    In time $\mathtt{T_o}|\dc\mathcal{F}|n^{O(1)}$, we can compute for each $F \in \dc\mathcal{F}$ and each $i \in [k]$ whether there exist sets $S_1,\ldots,S_i \in \mathcal{S}$ such that $\cup_{j=1}^{i}S_j \supseteq F$.   
\end{claim}

First, use Claim~\ref{clmmid} to compute for each $i$ and $F \in \mathcal{F}$ whether there exist $S_1,\ldots,S_i \in \mathcal{S}$ such that $\cup_{j=1}^{i}S_j \supseteq F$. Second, we use Claim~\ref{clmmid} with set family $\mathcal{F}' = \{ U \setminus F : F \in \mathcal{F} \}$ to compute for each $i$ and $F \in \mathcal{F}$ whether there exist $S_{k-i},\ldots,S_{k} \in \mathcal{S}$ such that $\cup_{j=i}^{k}S_j \supseteq U \setminus F$. Now we can output YES if and only if these conditions hold for some $F$ and $i$.

The run time of this algorithm is $(|\dc \mathcal{F}|+|\dc\mathcal{F}'|)\mathtt{T_o}\cdot n^{O(1)}$.
Since $|\dc\mathcal{F}| \leq (2-\eps)^n$ by assumption, it remains to upper bound $|\dc \mathcal{F}'|$. Let $a_\sigma$ be the number of sets in $\dc \mathcal{F}'$ of size $\sigma n$. Since a set in $\dc\mathcal{F}'$ must be a subset of $U \setminus F$ for some $F \in \mathcal{F}$, we have that
\[
 a_\sigma \leq  \sum_{F \in \mathcal{F}} \binom{|U \setminus F|}{\sigma n} \leq (2-\eps)^n\binom{n/2}{\sigma n}.
\]
On the other hand, we trivially have that $a_\sigma \leq \binom{n}{\sigma n}$. Hence, by~\eqref{eq:binstir} and its subsequent remark we have that
\begin{equation}\label{eq:bnd}
    a_{\sigma n} \leq \left( \min \{ 2^{h(\sigma)}, (2-\eps)2^{h(2\sigma)/2}\}  \right)^n.
\end{equation}
Now the lemma follows since $h(\sigma) < 1$ whenever $\sigma \neq \tfrac{1}{2}$, and if $\sigma=\tfrac{1}{2}$ then $h(2\sigma)=0$. More formally and quantitatively, one can use the inequality $h(\tfrac{1}{2}-x)=h(\tfrac{1}{2}+x)\leq 1-x^2$ that holds for all $x \in (0,\tfrac{1}{2})$.
\end{proof}

As we will see below, in several special cases of \textsc{Set Cover}, it is possible to construct a family $\mathcal{F}$ that \emph{witnesses} a solution of the \textsc{Set Cover} instance which is small enough for getting a non-trivially fast algorithm via Lemma~\ref{lem:crossmid}. The first example simply construct $\mathcal{F}$ at random by adding each set of size roughly $n/2$ to it with probability $1/2^{\Omega(n)}$:

\begin{theorem}[\cite{DBLP:conf/esa/Nederlof16}]\label{thm:sclarge}
For every $\sigma >0$, there exists $\eps >0$ such that any instance of \textsc{Set Cover} with $k \geq \sigma n$ can be solved with a Monte Carlo algorithm in $|\mathcal{S}|^{O(1)}(2-\eps)^n$.
\end{theorem}
\begin{proof}[Proof sketch]
Let $\eps_1$ be a parameter that only depends on $\sigma$ to be set later.
If the solution contains a set of size at least $\eps_1 n$, we can detect this solution in $|\mathcal{S}|^{O(1)}(2-\eps_1)^n$ time by guessing the set and solving the \textsc{Set Cover} instance induced by all elements not in the guessed set with Lemma~\ref{thmsc}. Hence we may assume from now on that all sets of size at least $\eps_1 n$ are not used in a solution.

Now our algorithm constructs $\mathcal{F} \subseteq 2^{U}$ by including in it each subset of $U$ with cardinality at least $n/2$ and at most $(1/2+\eps_2)n$ independently with probability $p=2^{-\sigma n}$, where $\eps_2$ is a parameter that only depends on $\sigma$ and is much smaller than $\eps_1$.
Subsequently it runs Lemma~\ref{lem:crossmid} with the set families $\cF$ and $\mathcal{S}$ and outputs whether it detected $S_1,\ldots,S_k \in \mathcal{S}$, $i \in [k]$ and $F \in\mathcal{F}$ such that $\cup_{j=1}^{i-1}S_j \supseteq F$ and $\cup_{j=i}^{k}S_j \supseteq U \setminus F$.

Ensuring that $\eps_2$ is much smaller than $\sigma$, it is not too hard to show that $\Pr[|\dc\mathcal{F}| \leq 2^{(1-\eps')n}]\geq 0.9$ for some $\eps' >0$ that only depends on $\sigma >0$ with a Markov bound and argument similar to the one used in~\eqref{eq:bnd}. Therefore this algorithm runs in the claimed time.

It remains to analyse the probability that the algorithm is correct. If the algorithm outputs YES, it is clear that the found sets $S_1,\ldots,S_k$ indeed witness a solution to the \textsc{Set Cover} instance. For the other direction, suppose $S_1,\ldots,S_k \subseteq \mathcal{S}$ are such that their union equals $U$. Arbitrarily pick $S'_i \subseteq S_i$, for every $i \in [k]$ such that $\cup_{i=1}^k S'_i=U$ and $S'_i$ and $S'_j$ are disjoint whenever $i\neq j$.
If we pick a random subset $L \subseteq [k]$ obtained by including each element of $[k]$ independently with probability $1/2$ to $L$, then we get by a Hoeffding bound that $\Pr[ n/2 \leq  \sum_{i \in L}|S_i| \leq (\tfrac{1}{2}+\eps_2)n ]\geq \Omega(1)$ since $|S_i| \leq \eps_1 n$ for each $i \in [k]$ and $\eps_1$ is much smaller than $\eps_2$ and $\sigma$. This means that the set
\[
    \mathcal{W} = \left\{ \bigcup_{i \in L} S_i : n/2 \leq  \sum_{i \in L}|S_i| \leq (\tfrac{1}{2}+\eps_2)n \right\}, 
\]
is of size $\Omega(2^k)=\Omega(2^{\sigma n})$, and hence the probability that $\mathcal{F}$ and $\mathcal{W}$ intersect is
\[
    1-(1 - p)^{|\mathcal{W}|} \geq 1-\exp(- p |\mathcal{W}|)=1-\exp(-\Omega(1)) = 1-\Omega(1).
\]
The proof follows since the algorithm clearly outputs YES if $\mathcal{F}$ and $\mathcal{W}$ intersect.
\end{proof}
A more special (but perhaps more natural) case of \textsc{Set Cover} is that where all sets are of bounded size. The following result was achieved by dynamic programming over `relevant' subsets.
\begin{theorem}[\cite{DBLP:conf/iwpec/Koivisto09}]
For every $d$, there exists an $\eps >0$ that only depends on $d$ such that any instance of \textsc{Set Cover} with sets of size at most $d$ can be solved in $O^*((2-\eps)^{n})$ time. 
\end{theorem}

We skip the proof and instead refer to the original paper~\cite{DBLP:conf/iwpec/Koivisto09} or a textbook treatment in~\cite[Section 3.4]{DBLP:series/txtcs/FominK10}.
In fact, the algorithm of Koivisto achieves a run time of $2^{(1-1/\Omega(d))n}$. This run time is the best known, and somewhat curiously, similar to fastest run time for $k$-CNF-SAT as discussed in Section~\ref{sec:sat}.

\subsection{Application to Bin Packing}
Another example in which the ideas behind Lemma~\ref{lem:crossmid} play a crucial role is a recent algorithm for the following problem. For $w:[n]\rightarrow \mathbb{N}$, we use the short hand notation $W(X)=\sum_{e \in X}w(e)$.

\defproblemu{\textsc{Bin Packing}}{Non-negative integers $w(1),\ldots,w(n),c,k$ }{A partition of $[n]$ into sets $S_1,\ldots,S_k$ such that $w(S_i)=\sum_{e \in S_i}w(e) \leq c$ for all $i$.}

\noindent Indeed, \textsc{Bin Packing} is a special case of \textsc{Set Cover} where $[n]=U$ and $\mathcal{S}$ consists of all subsets $S \subseteq [n]$ such that $w(S)\leq c$. Therefore, Theorem~\ref{thmsc} already solves this problem in $2^{n}n^{O(1)}$ time.
The special structure of the sets created in this instance can however be exploited:
\begin{theorem}[\cite{DBLP:journals/siamcomp/NederlofPSW23}]
\label{bp}
For every $k$ there exists an $\eps >0$ such that \textsc{Bin Packing} can be solved in $O^*((2-\eps)^n)$ time.
\end{theorem}

Since the proof of this theorem requires quite a number of (technical) ideas that are beyond the scope of this survey, we make two simplifying assumptions:
\begin{description}
    \item[] \textbf{A1} The instance is \emph{tight} in the sense that $w([n])=k\cdot c$,
    \item[] \textbf{A2} There is a solution $S_1,\ldots,S_k$ and integer $\ell$ such that $\sum_{j=1}^\ell |S_j| = n/2$.
\end{description}
Assumption \textbf{A1} is rather strong and, roughly speaking, in the original paper~\cite{DBLP:journals/siamcomp/NederlofPSW23} it is lifted it by rounding the integers in such a way that the slack of a solution in each bin (i.e. the quantity $c - w(S_i)$) becomes equal to $0$.
Assumption \textbf{A2} is somewhat more mild and can be dealt with in a way that is similar to the proof of Theorem~\ref{thm:sclarge}: Roughly speaking, if there is a solution $S_1,\ldots,S_k$ in which some $S_i$ is very large, we can detect the solution by other means. Otherwise, we can order the sets of the solution as $S_1,\ldots,S_k$ such that $\sum_{j=1}^\ell |S_j| = n/2$.

Given these assumptions, the following simple lemma immediately gives a strong indication towards improvements of the aforementioned $2^nn^{O(1)}$ time algorithm for \textsc{Bin Packing}
\begin{lemma}
    Let $w(1),\ldots,w(n),k,c$ be an instance of \textsc{Bin Packing} satisfying assumptions \textbf{A1} and \textbf{A2}, and let
    \begin{equation}\label{eqf}
        \mathcal{F} = \{ X \subseteq [n] : w(X) = c \cdot k / 2 \}.
    \end{equation}
    For every $\eps >0$ there exists an $\eps'>0$ such that if $|\mathcal{F}| \leq (2-\eps)^n$, then the instance of \textsc{Bin Packing} can be solved in $(2-\eps')^n$ time.
\end{lemma}
\begin{proof}[Proof sketch]
The set $\cF$ can be enumerated in $(2-\eps)^{n} n^{O(1)}$ time with standard methods	(as also outlined in Section~\ref{sec:sss}). Apply Lemma~\ref{lem:crossmid} with $\mathcal{S}$ consisting of all subsets $S \subseteq [n]$ such that $w(S) = c$, and $\mathcal{F}$ as defined in~\eqref{eqf}.
\end{proof}

But what if $|\mathcal{F}| \geq (2-\eps)^{n}$?
This looks like a rather special case, but its exact structure is not immediately clear.
For further analysis, let us define the \emph{maximum frequency} 
\begin{align*}
\beta(w) &= \max_{v}|\{ X\subseteq \{1,\ldots,n \} : w(X)=v \}|\\
\intertext{Thus, $|\mathcal{F}| \leq \beta(w)$. We also introduce the parameter \emph{number of distinct subset sums}:}\\
    |w(2^{[n]})| &= |\{w(X) : X \subseteq [n]\}|.
\end{align*}
The two parameters will be competing and lead us to a win/win strategy based on which of the two is small enough: It is fairly straightforward to solve \textsc{Bin Packing} in time $|w(2^{[n]})|^kn^{O(1)}$ with a variant of the pseudo-polynomial time algorithm for \textsc{Subset Sum}.\footnote{For $i=1,\ldots,n$ and $(c_1,\ldots,c_k) \in w(2^{[n]})^k$ define a table entry that stores whether there exists a partition $S_1,\ldots,S_k$ of $[i]$ such that $w(S_i)=c_i$.}

To get some intuition about why the parameters are competing, here are two extremal cases:
\[
\begin{aligned}
    (w_a(1),w_a(2),\cdots,w_a(n))&=(0,0,0,\cdots,0)&|w_a(2^{n})| &= 1,&\beta(w_a) &= 2^n,\\
    (w_b(1),w_b(2),\cdots,w_b(n))&=(1,2,4,\cdots,2^{n-1})&|w_b(2^{n})| &= 2^n,&\beta(w_b) &= 1.
\end{aligned}
\]

The following result in additive combinatorics shows that indeed the two parameters are competing and therefore gives the final ingredient for Theorem~\ref{bp}:
\begin{lemma}[\cite{DBLP:journals/siamcomp/NederlofPSW23}]
For every $\eps >0$ there exists an $\eps'>0$ such that $|w(2^{[n]})| \geq 2^{\eps n}$ then $\beta(w) \leq (2-\eps')^n$.
\end{lemma}
While the original proof from~\cite{DBLP:journals/siamcomp/NederlofPSW23} gave an $\eps'$ that was exponentially small in $\eps$, this was improved to a polynomial dependency in~\cite{Jain2021Anticoncentration}.

\subsection{Set Cover with Containers}
In this section we provide a new promise version of \textsc{Set Cover} that will be used in the next section to give a proof of Theorem~\ref{thm:kcolzam}. Our proof is different from the original proof~\cite{DBLP:conf/stoc/Zamir23}, but it gives a less general result. The advantage from presentation purposes of our alternative proof is however that we can use general results on \textsc{Set Cover} as a (somewhat natural) black box.

\defproblemu{\textsc{Set Cover with Containers}}{Universe $U$ of cardinality $n$, `container' subsets $C_1,\ldots,C_k \subseteq U$, and an oracle $\mathtt{A_o}$. Oracle $\mathtt{A_o}$ takes a subset of $U$ as input and outputs in $\mathtt{T_o}$ time\true or\false. The promise property is that if there are $S_1,\ldots,S_k$ with $\cup_{i=1}^k S_i = U$ and $\mathtt{A_o}(S_i)=\true$ for each $i\in\{1,\ldots,k\}$, then there are such $S_1,\ldots,S_k$ with the additional property that $S_i \subseteq C_i$ for each $i \in \{1,\ldots,k\}$.}{Are there $S_1,\ldots,S_k$ with $\cup_{i=1}^k S_i = U$ and $\mathtt{A_o}(S_i)=\true$ for each $i\in\{1,\ldots,k\}$?}

\begin{theorem}\label{containedsetcover}
For any $k$, there exist $\eps,\hat{\eps}>0$ such that \textsc{Set cover with containers} $C_1,\ldots,C_k$ satisfying $|C_i| \leq (1/2+\eps)n$ can be solved in $O((2-\hat{\eps})^n\cdot\mathtt{T_o})$ time.
\end{theorem}
In order to prove Theorem~\ref{containedsetcover}, we first prove a lemma about set families that is formulated in terms of bipartite graphs in order to use standard graph notation (i.e. $N(v)$ for the neighborhood of a vertex $v$, $d(v)$ for $|N(v)|$ and $N(X)$ for $\cup_{v \in X}N(v)$).
\begin{lemma}
\label{lem:annoying}
For any $k$, there exist $\eps, \eps' >0$ such that the following holds: Let $G=(A \cup B,E)$ be a bipartite graph with $|A|=k$ and $|B|=n$ such that $d(a) \leq (\tfrac{1}{2}+\eps)n$ for every $a \in A$, and let $w:A\rightarrow \mathbb{N}$.
Then there exists an $X \subseteq A$ such that $w(X) \geq w(A)/2$ and $|N(X)|\leq (1-\eps')n$.
\end{lemma}
Before we prove the lemma, we state a lemma that we will need in the proof (curiously, it is not clear to us whether there is a simpler way to prove Lemma~\ref{containedsetcover} without this lemma).
\begin{lemma}[\cite{franktardos}]
\label{franktardos}
For any vector $\mathbf{w}=(w_1,\ldots,w_k) \in \mathbb{Z}^k$,  there is a vector $\mathbf{w'}=(w' _1,\ldots,w'_k) \in \mathbb{Z}^k$ such that $||\mathbf{w'}||_\infty\leq 2^{O(k^3)}$ and $sign(\langle\mathbf{w},\mathbf{x}\rangle)=sign(\langle\mathbf{w'},\mathbf{x}\rangle)$ for each $\mathbf{x}\in\{-1,0,1\}^k$.
\end{lemma}
\begin{proof}[Proof of Lemma~ \ref{lem:annoying}]
For convenience, assume $k$ is even (the case with $k$ being odd can be reduced to the case with $k$ being even as we can increase $k$ by $1$ and add a set $C_{k+1}=\emptyset$ with $w_{k+1}=0$ ).

By Lemma~\ref{franktardos}, there is a function $w': A \rightarrow \{0,\ldots,W\}$ with  $W=2^{O(k^3)}$ such that for all $X \subseteq A$ it holds that
$w(X)\geq w(A)/2$ if and only if $w' (X)\geq w'(A)/2$: Indeed, this directly follows by interpreting the functions $w$ and $w'$ as vectors.

Call a vertex $v \in B$ \emph{bad} if $w'(N(v)) > w'(A)/2$, and call it \emph{good} otherwise. Since $w'(N(v))$ and $w'(A)$ are integers, we have that $w'(N(v))\geq w'(A)/2+\tfrac{1}{2}$ if $v$ is bad.
If we let $b$ denote the number of bad elements, we have that
\[
 b\left(w'(A)/2+\tfrac{1}{2}\right) \leq \sum_{a \in A} w' (a)d(a) \leq \sum_{a \in A} w'(a) (\tfrac{1}{2}+\eps)n = w'(A) (\tfrac{1}{2} + \eps)n.
\]
Therefore, if we set $\eps$ such that $\eps \leq 1/(4kW) = 1/2^{O(k^3)}$, we obtain that
\[
 b \leq \left(\frac{\tfrac{1}{2}+\eps}{\tfrac{1}{2}+1/(2w'(A))}\right)n \leq \left(\frac{\tfrac{1}{2}+1/(4w' (A))}{\tfrac{1}{2}+1/(2w'(A))}\right)n  \leq \left(1-\frac{1}{4w'(A)}\right)n.
\]
Therefore, there are at least $\frac{n}{4w'(A)}=n/2^{O(k^3)}$ vertices in $B$ that are good. By the pigeonhole principle there exist $Y \subseteq A$ such that there are at least $n/(2^{O(k^3)}2^k)$ good vertices $v$ with $N(v)=Y$. Since these vertices are good, $w'(Y) \leq w'(A)/2$. If we set $\eps' = 1/(2^{O(k^3)}2^k)$, the set $X$ from the lemma is obtained as $X = A\setminus Y$: It satisfies $|N(X)| \leq (1- 1/(2^{O(k^3)}2^k)|B|$ since all good vertices are not in $N(X)$, and $w'(X)=w'(A)-w'(Y)\geq w'(A)/2$ implies that $w(X) \geq w(A)/2$.
\end{proof}

\begin{proof}[Proof of Lemma~\ref{containedsetcover}]
Let $\eps$ and $\eps'$ be as given by Lemma~\ref{lem:annoying} after fixing $k$.
Iterate over all subsets $L \subseteq [k]$ and consider $C_L = \cup_{ i\in L} C_i$. If $|C_L| \leq (1-\eps')n$, then apply Lemma~\ref{lem:crossmid} with $\mathcal{F} = \binom{C_L}{\geq n/2}$ and $\mathcal{S}$ to detect whether there exist $S_1,\ldots,S_k$ and $i \in [k]$ such that $\cup_{j=1}^{i-1}S_j \supseteq F$ and $\cup_{j=i}^k S_j\supseteq U \setminus F$ for some $F\in\mathcal{F}$. Output\true if and only if any of these iterations detect such $S_1,\ldots,S_k$.

We have that $|\dc\mathcal{F}|\leq 2^{|C_L|} \leq 2^{(1-\eps')n}$ and therefore the algorithm of Lemma~\ref{lem:crossmid} runs in time $O((2-\hat{\eps})^{n})$ time for some $\hat{\eps}>0$.

To see the correctness of this algorithm, note it is always correct if it outputs\true. For the other direction, let $S_1,\ldots,S_k$ be a solution such that $S_i \subseteq C_i$ for each $i \in \{1,\ldots,k\}$. 
Apply Lemma~\ref{lem:annoying} to the graph $G$ with $B=U$ and for each $i=1,\ldots,k$ a vertex $a \in A$ with $N(a)=C_i$ and $w(a)=|S_i|$. We conclude from the lemma that there is a set $X \subseteq [k]$ such that $|\cup_{i \in L}C_i| \leq (1-\eps)n$ and $\sum_{i \in L}|S_i| \geq  n/2$. If we try this $L$, the set $\cup_{i \in L}S_i$ is in $\binom{C_L}{\geq n/2}$ and therefore the algorithm of Lemma~\ref{lem:crossmid} will output\true.
\end{proof}

\subsection{Application to Regular Graph Coloring}
We will now use Theorem~\ref{containedsetcover} to obtain an algorithm for \textsc{$k$-Coloring} that is significantly faster than the $O^*(2^n)$ time algorithm implied by Theorem~\ref{thmsc} in the special case that $k$ is a constant and the input graph is \emph{regular} (i.e. every vertex has the same number of neighbors).

To do so we use a family of contained as given by the following lemma, which is one of the most basic results of the `container method' in combinatorics. See also~\cite[Theorem 1.6.1]{alon2016probabilistic}. 

\begin{lemma}[\cite{sapozhenko2007number}]\label{lem:container}
Let $G=(V,E)$ be an $n$-vertex $d$-regular graph and $\eps >0$. Then one can construct in $\ell \cdot \poly(n)$ time a collection of subsets $C_1,\ldots,C_\ell \subseteq V$ such that
\begin{enumerate}
    \item $\ell \leq \binom{n}{\leq n/(\eps d)}$,
    \item for each $i=1,\ldots,\ell$ we have that $|C_i| \leq \frac{n}{\eps d}+\frac{n}{2-\eps}$, and
    \item each independent set of $G$ is contained in some $C_i$.
\end{enumerate}
\end{lemma}

\begin{theorem}[\cite{DBLP:conf/stoc/Zamir23}]\label{thm:kcolzam}
For every $k$, there exists an $\eps >0$ such that \textsc{$k$-Coloring} on regular graphs with $n$ vertices can be solved in $O((2-\eps)^n)$ time.
\end{theorem}
\begin{proof}
Let $\eps_0,\eps' >0$ be the constants given by Theorem~\ref{containedsetcover} that depend on $k$ such that \textsc{Set Covering with Containers} in which all containers $C_i$ are of size at most $(1/2+\eps_0)n$ can be solved in $(2-\eps')^n$ time.

Let $d_0 \geq 1/\eps_0^2$ be a constant that we will fix later, and let $G$ be $d$-regular.
If $d \leq d_0$ run the algorithm from Theorem~\ref{colbd}.
Otherwise, apply Lemma~\ref{lem:container} with $\eps_1:=\eps_0/2$.
We get containers $C_1,\ldots,C_\ell$ with
\[
    |C_i|\leq n\left(\frac{1}{\eps_1 d}+\frac{1}{2-\eps_1}\right)\leq n\left(\frac{1}{2}+\eps_1+\frac{1}{\eps_1 d}\right)\leq n\left(\frac{1}{2}+\eps_0/2+\frac{2}{\eps_0d} \right)\leq n\left(\frac{1}{2}+\eps_0\right).
\]
Now we guess the containers containing the independent set $S_1,\ldots,S_k$ that form the color classes of a $k$-coloring of $G$, if it exists. Since there are $\binom{\ell}{k}$ such options, the runtime therefore will be
\[
    \binom{\ell}{k} (2-\eps')^n  \leq \binom{n}{n/(\eps_1 d)}^k (2-\eps')^n = \left(2^{k\cdot h(1/(\eps_1 d))} (2-\eps')\right)^{n}.
\]
Since $h(p)$ tends to $0$ when $p$ tends to zero, we can pick $d$ as a function of $k$ (and $\eps$ and $\eps'$, but these are also implied by $k$) such that $2^{k\cdot h(1/(\eps_1 d))} (2-\eps') < 2$.

\end{proof}

\subsection{\textsc{Set Cover} versus Asymptotic Tensor rank}
In exciting new works, \cite{DBLP:journals/corr/abs-2310-11926,
DBLP:journals/corr/abs-2311-02774} it is shown that \textsc{Set Cover} in which all sets are bounded in size by a constant can be solved $O(1.89^n)$ time, if a certain family of $3$-dimensional tensors has small asymptotic tensor rank.
Similarly as for matrices, the rank $rk(\mathbf{T})$ of a tensor $\mathbf{T}$ is the minimal number $r=rk(\mathbf{T})$ of rank-$1$ tensors $\mathbf{T}_1,\ldots,\mathbf{T}_r$ such that $\sum_{i=1}^t \mathbf{T}_i$ (where the sum is such entry-wise). Here, a rank one tensor is the outer product of three vectors (whereas for a matrix, a matrix of rank $1$ is a matrix that can be written as the outer product of $2$ vectors). The asymptotic rank of a tensor $\mathbf{T}$ is defined as $\lim_{r \to \infty} rk(\mathbf{T}^{\otimes r})^{1/r}$.

Strassen~\cite{strassen1994algebra} conjectured that any tensor satisfying certain mild conditions has small tensor rank, and curiously it is currently open to find an explicit tensor family of $3$-dimensional tensors of strongly super-quadratic tensor rank. 

It remains to be seen whether this new connection can be used to point toward more evidence that even the specific tensors at hand do have large tensor rank (contradicting the conjecture of Strassen) by for example connecting it with Hypothesis~\ref{seth}, or for directly aiming at faster algorithms for \textsc{Set Cover}.
\section{Path Finding}
\newcommand{\matchMatrix}{\ensuremath{\mathbf{H}}}
\newcommand{\cutMatrix}{\ensuremath{\mathbf{S}}}
\newcommand{\partitionsetm}[1]{\ensuremath{\Pi_{\mathrm{m}}(#1)}\xspace}
\newcommand{\cuts}[1]{\ensuremath{\Pi_2(#1)}\xspace}

Another area in which much progress has been made in recent years is that of the algorithm design for finding long paths and cycles. Formally, we consider the following problem:

\defproblemu{\textsc{(Un)directed Hamiltonicity}}{An (un)directed graph $G$ on $n$ vertices.}{Does $G$ have an Hamiltonian cycle?}

Several breakthrough results were obtained, most prominently the $O(1.66^n)$ time algorithm for \textsc{Undirected Hamiltonicity}~\cite{DBLP:journals/siamcomp/Bjorklund14}.
This research started with parameterized algorithms for finding paths of length at least $k$~\cite{DBLP:conf/icalp/Koutis08}, which was in turn inspired by a much earlier series of papers for determining whether a graph has a perfect matching in an algebraic manner via determinants. We will outline this earliest work in order to build upon it subsequently:
\begin{definition}[Tutte Matrix~\cite{Tutte47}] Let $G=(V,E)$ be a graph with linear ordering $\prec$ on $V$, let $\mathbb{F}$ be a field and for every $i<j$ let $x_{ij} \in \mathbb{F}$. Define
\[
	\mathbf{A}^{(x)}_G[i,j] =
	\begin{cases}
	x_{ij}	& \text{if } \{i,j\} \in E  \text{ and } i\prec j, \\
	-x_{ji}	& \text{if } \{i,j\} \in E  \text{ and } j \prec i, \\	
	0			& \text{otherwise }.
	\end{cases}
\]
\end{definition}
For a set $V$, we let $\partitionsetm{V}$ denote the family of all perfect matchings of the complete graph with vertex set $V$.
\begin{lemma}[\cite{Tutte47}]\label{lem:dettut} The determinant $\det(\mathbf{A}^{(x)}_G)$ is the polynomial in variables $x_{i,j}$ satisfying
\[
    \det(\mathbf{A}^{(x)}_G) = \sum_{M \in \partitionsetm{V}} \prod_{\substack{\{i,j\} \in M \\ i \prec j}} x^2_{ij}.
\]
\end{lemma}
Let us refer to the polynomial $\det(\mathbf{A}^{(x)}_G)$ as $P_G(x)$.
We first briefly describe an application of this lemma from~\cite{lovasz1979determinants} to get a fast randomized algorithm for determining whether $G$ has a perfect matching:
It is easy to see that $P_G$ is the zero polynomial if and only if $G$ does not have a perfect matching.
Since $P_G$ is a polynomial of degree at most $n$, we can use the following lemma to check whether $G$ has a perfect matching:

\begin{lemma}[Polynomial Identity Testing, \cite{DBLP:journals/ipl/DemilloL78, DBLP:journals/jacm/Schwartz80,DBLP:books/daglib/0000420}]\label{lem:pit}
Let $\mathbb{F}$ be a field and let $P(x_1,\ldots,x_z)$ be a non-zero polynomial on $z$ variables with values in $\mathbb{F}$ of degree at most $d$. If $r_1,\ldots,r_z \in \mathbb{F}$ are picked independently and uniformly and random, then $\Pr[P(r_1,\ldots,r_z)=0]\leq d/|\mathbb{F}|$.
\end{lemma}

In particular, fix $\mathbb{F}$ to be field of size at least $2n$ (which is at least twice the degree of $P_G$), and replace the variables $x_{ij}$ with random elements from $\mathbb{F}$, evaluate $P_G$ in $n^{\omega}$ time\footnote{We let $\omega$ denote the smallest constant such that $n$ by $n$ matrices can be multiplied in $n^{\omega + o(1)}$ time, $2 \leq \omega \leq 2.73$. It is well-known that the determinant of an $n \times n$ matrix can be computed in $n^{\omega +o(1)}$ time.} with a Gaussian-elimination based algorithm and output\true if and only if it evaluates to a non-zero number.

Since~\cite{DBLP:journals/siamcomp/Bjorklund14} there have been several algorithms for \textsc{(Un)directed Hamiltonicity} by evaluating the sum of (an exponential number of) determinants. We will survey some selected approaches, but skip a thorough discussion of~\cite{DBLP:journals/siamcomp/Bjorklund14} since several write-ups already exist in textbooks~\cite{DBLP:books/sp/CyganFKLMPPS15} or surveys (such as~\cite{DBLP:journals/cacm/FominK13, DBLP:journals/cacm/KoutisW16}).

\subsection{Hamiltonicity via Matrix Factorizations}
We first give a relatively simple $O^*(2^n)$ time algorithm illustrating the main idea behind a faster algorithm that we will present afterwards. For this, we first define the following two matrices:

\begin{definition}[Matchings Connectivity matrix] For even $t \geq 2$, define $\matchMatrix_t \in \{0,1\}^{\partitionsetm{[t]} \times \partitionsetm{[t]}}$ as
	\[ \matchMatrix_t[A,B] =
	\begin{cases}
	1, &\text{if $A \cup B$ is a Hamiltonian Cycle},\\
	0, &\text{otherwise}.
	\end{cases}
	\]	
\end{definition}

For a set $V$, we let $\cuts{V}$ denote the family of all \emph{cuts of $V$}, i.e. unordered partitions of $V$ into two blocks.

\begin{definition}[Split matrix]
	A matching $A \in \partitionsetm{[t]}$ is \emph{split} by a cut $C \in \cuts{[t]}$ if every edge of $A$ is either contained in $C$ or is disjoint from $C$. For even $t \geq 2$, define $\cutMatrix_t \in \{0,1\}^{\partitionsetm{[t]} \times \cuts{[t]}}$ as
	\[ \cutMatrix_t[A,C] =
	\begin{cases}
	1, &\text{if $A$ is split by $C$},\\
	0, &\text{otherwise}.
	\end{cases}
	\]	
\end{definition}

\paragraph{A $O^*(2^n)$ time algorithm for \textsc{Undirected Hamiltonicity}.}
Let $A,B \in \partitionsetm{[t]}$.
It is easy to see that the number of cuts $C \in \cuts{[t]}$ that split both $A$ and $B$ simultaneously is $2^{k-1}$, where $k$ is the number of connected components of the graph $([t],A \cup B)$. Hence, since $2^{k-1}$ is odd if and only if $k=1$, we have over any field of characteristic $2$ that $\mathbf{H}_t = \mathbf{S}_t \mathbf{S}_t^\intercal$. Thus, let $\mathbf{p},\mathbf{q}$ denote
\[
    \mathbf{p}[A] = \begin{cases}
        \displaystyle \prod_{\substack{\{i,j\} \in A \\ i \prec j}} x^2_{ij}, & \text{if } A \in \partitionsetm{[t]} ,\\
 0, & \text{otherwise},
    \end{cases} \ \qquad \qquad \ \ \
       \mathbf{q}[A] = \begin{cases}
        \displaystyle\prod_{\substack{\{i,j\} \in A \\ i \prec j}} y^2_{ij}, & \text{if } A \in \partitionsetm{[t]},\\
 0, & \text{otherwise}.
    \end{cases}
\]
Now by the earlier observation we have that, in any field of characteristic two, $\mathbf{p}^{\intercal} \mathbf{H}_t \mathbf{q}$ is the zero polynomial if and only if $G$ has no Hamiltonian cycle. It follows that the existence of a Hamiltonian cycle can be checked in $O^*(2^n)$ time by plugging in random values from the field $GF(2^{k})$ for $k=\log_2 8n$ into $x_{ij},y_{ij}$ and evaluating the following polynomial (which has degree at most $4n$) in the straightforward way:
\[
\mathbf{p}^{\intercal} \mathbf{H}_t \mathbf{q} = (\mathbf{p}^\intercal\mathbf{S}_t)( \mathbf{S}_t^\intercal\mathbf{q}) = 
\sum_{C \in \cuts{[n]}}
\det(\mathbf{A}^{(x)}_{G[C]}) \det(\mathbf{A}^{(x)}_{G[V \setminus C]})
\det(\mathbf{A}^{(y)}_{G[C]}) \det(\mathbf{A}^{(y)}_{G[V \setminus C]}).
\]

\paragraph{A $O^*(3^{n/2})$ time algorithm for \textsc{Undirected Hamiltonicity}.}
Now we see a faster algorithm that uses the same blueprint as the previous algorithm, but to get a faster algorithm we use the following more efficient factorization of $\matchMatrix_t$:

\begin{lemma}[Narrow Cut Factorization, \cite{DBLP:conf/stoc/Nederlof20}]\label{lem:narrowcutfact} Let $t \geq 2$ be an even integer.
There exists a polynomial-time computable function $C:\{0,1,2\}^{t/2-1} \rightarrow \cuts{[t]}$ such that, if we let $\mathcal{C}_t = \{ C(x) : x \in \{0,1,2\}^{t/2-1}\}$ then, over a field $\mathbf{F}$ of characteristic $2$ we have
\[
	\matchMatrix_t = \cutMatrix_t[\cdot,\mathcal{C}_t]\cdot
 \begin{pmatrix}
0 & 1 & 1\\
1 & 0 & 1\\
1 & 1 & 0
\end{pmatrix}^{\otimes t/2-1}
	\cdot
	(\cutMatrix_t[\cdot,\mathcal{C}_t])^\intercal.
	\]
\end{lemma}
Curiously, the rank of $\matchMatrix_t$ over fields of characteristic $2$ is equal to $2^{t/2-1}$ (see~\cite{DBLP:journals/jacm/CyganKN18}), but we are not aware of narrower factorizations in terms of $\cutMatrix$ and the narrower factorizations from~\cite{DBLP:journals/jacm/CyganKN18} seem harder to combine with the algorithmic approach outlined here.
\begin{theorem}\label{thm:undham}
	\textsc{Directed Hamiltonicity} in bipartite graphs and \textsc{Undirected Hamiltonicity} can be solved in $O^*(3^{n/2})$ time.
\end{theorem}
These algorithmic results simplify and improve a similar approach from~\cite{DBLP:journals/jacm/CyganKN18}, but are in turn inferior to the results from~\cite{DBLP:journals/siamcomp/Bjorklund14} and~\cite{DBLP:conf/icalp/BjorklundKK17}.
We nevertheless present the result here since it already follows from a combination of Lemma~\ref{lem:narrowcutfact} with standard methods.

\begin{proof}[Proof sketch of Theorem~\ref{thm:undham}]
We focus on the second item of Theorem~\ref{thm:undham}.\footnote{The first item can be proved in similar fashion by only using $x$-variables for arcs in one direction and $y$-arcs for all arcs in the other direction.}
The algorithm is outlined in Algorithm~\ref{alg:undHam}.
\begin{algorithm}
	\caption{Undirected Hamiltonicity via the Narrow Cut Factorization.}
	\label{alg:undHam}
	\begin{algorithmic}[1]
		\REQUIRE $\mathtt{undirectedHamiltonicity}(G=(V,E))$
		\ENSURE \true with probability at least $1/2$ if $G$ is Hamiltonian; \false otherwise.
		\STATE For each $\{i,j\} \in E$ with $i<j$ pick $x_{ij}, y_{ij}  \in_R  GF(2^k)$, where $k= \log_2 8n$
		\FOR{$a \in \{0,1,2\}^{n/2-1}$}\label{lin:loopund}
		\STATE $\mathbf{l}[a] \gets \det(\mathbf{A}^{(x)}_{G[C(a)]}) \cdot \det(\mathbf{A}^{(x)}_{G[V \setminus C(a)]})$\label{lin:det1}
		\STATE $\mathbf{r}[a] \gets \det(\mathbf{A}^{(y)}_{G[C(a)]}) \cdot \det(\mathbf{A}^{(y)}_{G[V \setminus C(a)]})$\label{lin:det2}
		\ENDFOR
		\STATE $res \gets \mathbf{l}^\intercal \cdot   \begin{pmatrix}
0 & 1 & 1\\
1 & 0 & 1\\
1 & 1 & 0
\end{pmatrix}^{\otimes n/2-1} \cdot  \mathbf{r} $ \label{lin:kronecker}
		\LineIfElse{$res\neq 0$}{$\mathbf{return}$\true{}}{$\mathbf{return}$\false{}}
	\end{algorithmic}
\end{algorithm}
Note it takes $O^*(3^{n/2})$ time because there are $3^{n/2-1}$ iterations of the loop at Line~\ref{lin:loopund}, the determinants on Lines~\ref{lin:det1},~\ref{lin:det2} are computed in polynomial time with standard algorithms, and the vector-matrix-vector product product on Line~\ref{lin:kronecker} can be computed in $O^*(3^{n/2})$ using Yates' algorithm (Lemma~\ref{lem:yates}) and an inner-product computation.
For correctness, let us denote the $(3\times 3)$-matrix of Line~\ref{lin:kronecker} by $\mathbf{Q}$. Notice that the output $res$ of the algorithm is an evaluation of the polynomial $P(x,y)$ at random points $x,y$ where we have that $P(x,y)$ equals
\begin{align*}
&= \sum_{a,b \in \{0,1,2\}^{n/2-1}}\left(\prod_{i=1}^{n/2-1}\mathbf{Q}[a_i,b_i]\right)\det\left(\mathbf{A}^{(x)}_{G[C(a)]}\right) \det\left(\mathbf{A}^{(x)}_{G[V \setminus C(a)]}\right)
\det\left(\mathbf{A}^{(y)}_{G[C(b)]}\right) \det\left(\mathbf{A}^{(y)}_{G[V \setminus C(b)]}\right)\\
	%
	&\ \algcomment{By Lemma~\ref{lem:dettut}}\\
	&= \sum_{a,b \in \{0,1,2\}^{n/2-1}}\left(\prod_{i=1}^{n/2-1}\mathbf{Q}[a_i,b_i]\right) \left( \sum_{\substack{ M_1 \in \partitionsetm{V}\\ C(a) \text{ splits } M_1 }} \prod_{\substack{\{i,j\} \in M_1 \\ i \prec j}} x^2_{ij} \right)
	\left( \sum_{\substack{ M_2 \in \partitionsetm{V}\\ C(b) \text{ splits } M_2 }} \prod_{\substack{\{i,j\} \in M_2 \\ i \prec j}} y^2_{ij} \right)\\
	&= \sum_{M_1, M_2 \in \partitionsetm{G}}
	\left(
		\sum_{a,b \in \{0,1,2\}^{n/2-1}} [C(a) \text{ splits } M_1]\left(\prod_{i=1}^{n/2-1}\mathbf{Q}[a_i,b_i]\right) [C(b) \text{ splits } M_2]\right)\cdot\\ 
	 &\hspace{4em}\Bigg(\prod_{\substack{\{i,j\} \in M_1 \\ i \prec j}} x^2_{ij} \Bigg) \Bigg(\prod_{\substack{\{i,j\} \in M_2 \\ i \prec j}} y^2_{ij} \Bigg)\\
	 &\ \algcomment{By Lemma~\ref{lem:narrowcutfact}}\\
	 &= \sum_{M_1, M_2 \in \partitionsetm{G}}
	 \matchMatrix_t[M_1,M_2] 
	 \Bigg(\prod_{\substack{\{i,j\} \in M_1 \\ i \prec j}} x^2_{ij} \Bigg) \Bigg(\prod_{\substack{\{i,j\} \in M_2 \\ i \prec j}} y^2_{ij} \Bigg).
\end{align*}
Since $P(x,y)$ has degree at most $4n$, the correctness follows by Lemma~\ref{lem:pit}
\end{proof}

\subsection{Directed Hamiltonicity}
Now we outline the approach towards the following theorem
\begin{theorem}\label{thm:dh}[\cite{DBLP:conf/icalp/BjorklundKK17}]
    There is an algorithm that given a $n$-vertex directed graph $G=(V,E)$ and prime number $p$, counts the number of Hamiltonian cycles of $G$ modulo $p$ in time $2^{n\left(1-\frac{1}{O(p \log p)}\right)}$.
\end{theorem}
A stronger version appeared in~\cite{DBLP:conf/icalp/BjorklundKK17}, but we slightly simplified the statement.

The curious situation is that, despite this algorithm, there is still no known algorithm to solve \textsc{Directed Hamiltonicity} in $O^*((2-\eps)^n)$ time, for some $\eps >0$. For many algorithms, including the two previous algorithms from this section, algorithms that count the number of solutions modulo a prime number $p$ can be extended with Lemma~\ref{lem:pit} to solve problem of detecting a solution or they can solve a weighted modular counting version of the problem to which the decision version can be reduced with the \emph{isolation lemma}~\cite{DBLP:journals/combinatorica/MulmuleyVV87}. The algorithm behind~\ref{thm:dh} however explicitly relies on the fact that many intermediate computations result in value $0$ (and therefore can be skipped), which complicates the aim to ensure that solutions do not cancel each other out.

Let $G$ be a $n$-vertex directed multigraph. Since $G$ is a multi-graph, the set of edges $E(G)$ of $G$ is a multi-set. Fix two vertices $s,t \in V(G)$, and assume there is exactly one edge from $t$ to $s$ (if there are more, the approach can be easily adjusted by multiplying the outcome with the multiplicity of the edge $(t,s)$).
Let $\mathbf{A}_G$ be the adjacency matrix of $G$, so $\mathbf{A}_G[v,w] \in \mathbb{Z}_{\geq 0}$ describes the number of arcs $(v,w) \in E(G)$. We assume $\textbf{A}_G[v,v]=0$, i.e. the graph has no loops. Denote $d_G(w)=\sum_{v \in V}\mathbf{A}_G[v,w]$ for the in-degree of a vertex, and let $\mathbf{D}_G$ be the diagonal matrix with entry $\mathbf{D}_G[v,v]=d_{G}(v)$ for every vertex $v$.
The \emph{Laplacian} $\mathbf{L}_G$ is defined as $\mathbf{D}_G-\mathbf{A}_G$.
Also define the incidence matrices
\[
    \mathbf{I}_G[u,(v,w)] =
    \begin{cases}
        1, & \text{if } u=w,\\
        0, & \text{otherwise },
    \end{cases}
    \qquad 
    \text{and}
    \qquad
    \mathbf{O}_G[u,(v,w)] =
    \begin{cases}
        1, & \text{if } u=v,\\
        0, & \text{otherwise }.
    \end{cases}
\]
Note that
\begin{equation}\label{eq:factor}
    \begin{aligned}
    \left((\mathbf{I}_G-\mathbf{O}_G) \cdot \mathbf{I}_G^\intercal \right)[u,x] &= \sum_{(v,w) \in E(G)} ([u=w]-[u=v])([w=x])\\
    &=    \begin{cases}
    	|\{(v,w) \in E(G): u = w\}|, & \text{if } u=x,\\
    	-|\{(v,w) \in E(G): u =v,w=x\}|, & \text{otherwise}
    \end{cases}\\
    &=    \begin{cases}
        d_{G}(u), & \text{if } u=x,\\
        -\mathbf{A}_G[u,x], & \text{otherwise}
    \end{cases}\\
    &= (\mathbf{D}_G-\mathbf{A}_{G})[u,x] = \mathbf{L}_G[u,x].
    \end{aligned}
 \end{equation}
 Let $\mathbf{L}^{-s}_G$ denote the matrix obtained from $\mathbf{L}_G$ by removing the row and column indexed by $s$, and let $\mathbf{I}^{-s}_G$ and $\mathbf{O}^{-s}_G$ denote the matrices obtained by removing the row indexed by $s$ from respectively $\mathbf{I}^{-s}_G$ and $\mathbf{O}^{-s}_G$. By~\eqref{eq:factor} we have that
 \[
    \mathbf{L}^{-s}_G = (\mathbf{I}^{-s}_G-\mathbf{O}^{-s}_G)  (\mathbf{I}^{-s}_G)^\intercal.
 \]
We call a subset $X \subseteq E(G)$ an \emph{out-branching} if $X$ is a rooted spanning tree with all arcs directed away from the root.
 
\begin{lemma}\label{lem:factor}
    If $X \in \binom{E(G)}{n-1}$ we have that
    \begin{equation}\label{cbterm}
        \det\left((\mathbf{I}^{-s}_G-\mathbf{O}^{-s}_G)[\cdot,X]\right) \cdot \det\left(\mathbf{I}^{-s}_G[\cdot,X]\right) = [X \text{ is an out-branching rooted at $s$}].
    \end{equation}
\end{lemma}
\begin{proof}
We distinguish four cases:
\begin{itemize}
    \item If $(v,s) \in X$ for some $v \in V(G)$, then $\det(\mathbf{I}^{-s}_G[\cdot,X])=0$, since the column in $\mathbf{I}^{-s}_G[\cdot,X]$ indexed by $(v,s)$ consists of only zeroes
    \item If there are two distinct edges $(u,w),(v,w) \in X$ then $\det(\mathbf{I}^{-s}_G[\cdot,X])=0$ since some vertex $x$ has no incoming edges in $X$ and its corresponding row consists of only zeroes
    \item If $X$ forms a cycle, this cycle cannot pass through $s$ and hence the columns of $\mathbf{I}^{-s}_G-\mathbf{O}^{-s}_G$ that are indexed by the vertices in the cycle sum to $0$, implying $\det\left((\mathbf{I}^{-s}_G-\mathbf{O}^{-s}_G)[\cdot,X]\right)=0$
    \item Otherwise $X$ is an out-branching rooted at $s$. Then $\mathbf{I}^{-s}_G[\cdot,X]$ is a permutation matrix and therefore its determinant is $sgn(\sigma)$, where $\sigma$ maps each vertex to its unique incoming arc in $X$. But $\sigma$ is also the only permutation contributing to the determinant of $\mathbf{I}^{-s}_G-\mathbf{O}^{-s}_G$, and it contributes a factor $sgn(\sigma)$ as well. Thus~\eqref{cbterm} reduces to $sgn(\sigma)^2=1$.
\end{itemize}
\end{proof}
Now we have by Lemma~\ref{lem:factor}  and the Cauchy-Binet formula\footnote{The Cauchy-Binet formula states the following: if $\mathbf{A}$ is an $a\times b$ matrix and $\mathbf{B}$ is a $b\times a$ matrix, then $\det(\mathbf{A}\mathbf{B})=\sum_{X \subseteq \binom{[b]}{a}} \det(\mathbf{A}[\cdot,X])\det(\mathbf{B}[X,\cdot])$.} that
\begin{equation}
\label{eq:ob}
\begin{aligned}
   \det (\mathbf{L}^{-s}_G)  &=  \det ( (\mathbf{I}^{-s}_G-\mathbf{O}^{-s}_G)  (\mathbf{I}^{-s}_G)^T )\\
   &= \sum_{X \subseteq E(G)} \det\left((\mathbf{I}^{-s}_G-\mathbf{O}^{-s}_G)[\cdot,X]\right) \cdot \det\left(\mathbf{I}^{-s}_G[\cdot,X]\right)\\
    &= |\{ B: B \text{ is an outbranching of $G$ rooted at $s$}\}|.
\end{aligned}
\end{equation}

\begin{lemma} Let $out(F)$ denote all edges with starting point in $F$ and let $G-out(F)$ denote the graph obtained from $G$ by removing all edges $out(F)$. The number of Hamiltonian cycles in $G$ containing the arc $(t,s)$ equals
\begin{equation}\label{dethc}
    \sum_{F \subseteq V(G) \setminus \{t\}} (-1)^{|F|} \det(\mathbf{L}^{-s}_{G-out(F)}).
\end{equation}
\end{lemma}
The proof combines~\eqref{eq:ob} with inclusion-exclusion: 
\begin{proof}
By~\eqref{eq:ob}, we can rewrite~\eqref{dethc} into
\[
    \sum_{B} \sum_{F \subseteq V(G)\setminus \{t\}} (-1)^{F} [B \cap out(F)=\emptyset] = \sum_{B} \sum_{F \subseteq sinks(B)\setminus \{t\}} (-1)^{F},
\]
where the sums run over all out-branchings $B$ of $G$ rooted at $s$ and $sinks(B)$ denotes $\{v \in V(G): \forall w \in V(G): (v,w) \notin B\}$. Since every non-empty set has equally many odd-sized subsets as it has even-sized subsets, only sets $B$ with one sink (being $t$) contribute to~\eqref{dethc}. These contribute exactly one to~\eqref{dethc} and since these are exactly the Hamiltonian paths from $s$ to $t$ visiting all vertices, and hence (after addition of the arcs $(t,s)$) Hamiltonian cycles containing the arc $(t,s)$, the lemma follows.
\end{proof}
So what do we gain with computing the number of Hamiltonian cycles via~\ref{dethc}, since it still consists of $2^n$ summands? The point is that for many summands $F$ the term $\det(\mathbf{L}^{-s}_{G-out(F)})$ will be equal to $0$ (modulo a small number $p$).
In particular, on any row in $\mathbf{L}^{-s}_{G-out(F)}$ corresponding to a vertex in $F$, all entries will be zero except possibly the diagonal entry. Thus, if we work modulo a small $p$, if this diagonal entry would be zero as well (modulo $p$), then in fact the matrix does not have full rank and hence determinant is equal to $0$ (modulo $p$). But how to ensure that the diagonal entry is equal to $0$ modulo $p$?

Note that the number of arcs from $t$ to $v$ for $v\neq s$ does not matter at all for the outcome of the above algorithm.
So we could as well add a (uniformly, independently chosen) random number $r_v$ of edges from $t$ to each vertex $v \neq s$.

Observe that $\det(\mathbf{L}^{-s}_{G-out(F)})$ is equal to zero modulo $p$ whenever $\sum_{u \in V\setminus (F \cup \{t\})}\mathbf{A}[u,v]+r_v$ is equal to $0$ modulo $p$ for some $v \in F$.
Thus the probability that a summand $F$ contributes to~\eqref{dethc} is $(1-1/p)^{|F|}$. The algorithm behind Theorem~\ref{thm:dh} now enumerates a superset of these contributing terms in the claimed time bound and afterwards uses the enumerated list to evaluate~\eqref{dethc} in the direct manner. We skip details on the procedure that enumerates the contributing terms, and refer to the original paper~\cite{DBLP:conf/icalp/BjorklundKK17} for details.

\subsection{\textsc{Traveling Salesperson Problem}}
If we extend the \textsc{Undirected Hamiltonicity} problem with edge weights, we arrive at the following well known problem

\defproblemu{\textsc{Traveling Salesperson Problem} (TSP)}{A undirected graph $G=(V,E)$, distances $w:E \rightarrow \mathbb{N}$}{Find a Hamiltonian cycle $C$ of $G$ that minimizes $w(C)$.}

The algorithms from the previous section extend to this problem at some cost: For example, in~\cite{DBLP:journals/siamcomp/Bjorklund14} a write-up is given of a $O^*(1.66^n W)$ time algorithm for TSP, which also gives a $O^*(1.66^n /\eps )$ time $(1+\eps)$-approximation using standard rounding tricks. Nevertheless, it is interesting to see whether the pseudo-polynomial factor $W$ can be avoided in this run time. 
Especially, since the algebraic algorithms discussed before seem impossible to solve this optimization variant exactly without incurring this pseudo-polynomial overhead in the run time.

As such, the natural question whether the natural $O^*(2^n)$ time dynamic programming algorithm by Bellman~\cite{bellman1962dynamic}, Held and Karp~\cite{held1962dynamic} can be improved remains elusive.
In the model of tropical circuits (modeling, to some extent, dynamic programming algorithms), it is even shown that no faster algorithm exists~\cite{DBLP:journals/jacm/JerrumS82} (see also the recent textbook~\cite{jukna2023tropical}).

Faster algorithms were given for graphs of bounded degree~\cite{DBLP:journals/talg/BjorklundHKK12}, graphs of small pathwidth and treewidth~\cite{DBLP:journals/iandc/BodlaenderCKN15}, and (assuming $\omega=2$, where $\omega$ is the smallest number such that $q \times q$ matrices can be multiplied in $q^{\omega+o(1)}$ time) for bipartite graphs~\cite{DBLP:conf/stoc/Nederlof20}.

It was shown in~\cite{DBLP:journals/siamcomp/GurevichS87} that TSP can be solved in $4^n n^{O(\log n)}$ time and poly space (see e.g. \cite[Section 10.1]{DBLP:series/txtcs/FominK10}), but it is not known whether it can be solved faster $O^*((4-\eps)^n)$ time and polynomial space, for some $\eps >0$.

\section{Subset Sum}\label{sec:sss}

Another computational problem that saw exciting progress from the perspective of fine-grained complexity in the last decade or so is the following:

\defproblemu{\textsc{Subset Sum}}{A weight function $w:[n] \rightarrow [W]$, and a target integer $t$.}{Is there a subset $S \subseteq [n]$ such that $w(S)=t$.}

Here we use the notation from earlier sections that $w(X)=\sum_{e \in X}w(e)$ and if $\mathcal{F} \subseteq 2^{[n]}$ is a set family then we also denote $w(\mathcal{F}) = \{w(F): F \in \mathcal{F}\}$.

Since there is a reduction from the more general \textsc{Knapsack} problem to \textsc{Subset Sum} in the regime that $n$ is small (see~\cite{DBLP:conf/mfcs/NederlofLZ12}), all remarks below apply to the (arguably, more natural and central) \textsc{Knapsack} problem as well. But, for the sake of brevity, we restrict our attention to \textsc{Subset Sum}.

\subsection{Meet in the Middle and The Representation Method}

An elegant algorithm that solves \textsc{Subset Sum} faster than the completely na\"ive $O^*(2^n)$ time algorithm was already presented 50 years ago~\cite{DBLP:journals/jacm/HorowitzS74}.
In this algorithm, and the one that we will subsequently discuss, a central role will be played by two ``lists''.

In the first approach, these lists are defined as follows:
Partition the set $[n]$ into two sets $L,R$ of size $n/2$ each (assuming for convenience here and later that $n$ is a multiple of $4$, and hence even), and define
\begin{equation}\label{eq:ss}
    \mathcal{L} := 2^{L}, \qquad \mathcal{R} := 2^{R}.
\end{equation}

Now the algorithm is as follows: First, enumerate and sort the numbers in $w(\mathcal{R})$, and for each $X \in \mathcal{L}$ we check with binary search whether $t-w(X) \in w(\mathcal{R})$. If such element exists we can output\true{} (since we can take $S=X \cup Y$ where $Y$ is the subset of $R$ satisfying $w(Y)=t-w(X)$, as $X$ and $Y$ are disjoint). If no such element exists in any iteration we can output\false (since, if a set $S$ exists with $w(S)=t$, then we would have detected it at the iteration with $X = S \cap L$).
Since binary search runs in $\log |\mathcal{R}|$ time, this procedure runs in $O^*(2^{n/2})$ time.

A notable open question is whether this can be improved to, say, $O^*(2^{0.4999n})$ time.
In~\cite{DBLP:conf/eurocrypt/Howgrave-GrahamJ10} surprising progress was made on this in the context of random instances (as opposed to worst-case analysis) with an elegant method called \emph{the representation method}. The algorithm of~\cite{DBLP:conf/eurocrypt/Howgrave-GrahamJ10} is outlined in Algorithm~\ref{alg:repmet}.
To state what it exactly achieves, we need the following definition:

\begin{definition}
 A \emph{pseudo-solution} is a pair $(X,Y) \in \binom{[n]}{n/4} \times \binom{[n]}{n/4}$ such that $w(X)+w(Y)=t$.
\end{definition}

We make the following assumption (called \textbf{A1}-\textbf{A3}):
\begin{itemize}
\item[\textbf{A1}] The number of \emph{pseudo-solutions} is at most $2^{0.9n}$.
\end{itemize}
\noindent Moreover, if the instance is a YES-instance, then there exists a set $S \subseteq [n]$ such that $w(S)=t$,
\begin{itemize}
\item[\textbf{A2}] $|S|=n/2$, and
\item[\textbf{A3}] $|w(2^S)|=2^{n/2}$.
\end{itemize}

It can be proved that, if all $w(1),\ldots,w(n)$ are picked uniformly and independently at random from $[2^n]$ and $t:= w([n/2])$, then Assumptions \textbf{A1}-\textbf{A3} indeed hold with high probability.

The algorithm is depicted in Algorithm~\ref{alg:repmet}. The idea behind this algorithm is as follows: If we would define $\mathcal{L}=\mathcal{R}=\binom{[n]}{n/4}$, then ratio of the number of pairs $(X,Y) \in \mathcal{L}\times\mathcal{R}$ that witness the solution (i.e. $X$ and $Y$ are disjoint and $X \cup Y =S$) divided by the list size $\binom{n}{n/4}$ is
\[
	\binom{|S|}{n/4} / \binom{n}{n/4} \geq \tfrac{1}{n} 2^{\left(h\left(\tfrac{1}{2}\right)/2-h\left(\tfrac{1}{4}\right)\right)n}\geq 2^{-0.32n},
\]
where we use \eqref{eq:binstir} and its subsequent remarks in the first inequality.
This ratio is larger than the analogous ratio $1/ 2^{n/2}$ of $\mathcal{L}$ and $\mathcal{R}$ as defined in~\ref{eq:ss}. This can be leveraged by sampling one witness by picking a random prime $p$ and guessing $t_L \equiv_p w(X)$, the crux being that this single guess also determines $w(Y)$ modulo $p$ since $w(Y) \equiv_p t- t_l$.

Line~\ref{lin:clist} can be implemented to run in time $O^*(2^{0.45n} + |\mathcal{L}|+|\mathcal{R}|)$ as follows: Create a directed graph $G=(\{0,\ldots,n\} \times \mathbb{Z}_p \times \{0,\ldots,n/4\}, A)$ where we have for each $0\leq i \leq n-1$, $j \in \mathbb{Z}_p$, $k \in \{0,\ldots,n/4\}$ arcs
\[
((i,j,k),(i+1,j,k)), \quad \text{ and } \quad  ((i,j,k),(i+1,j+w(i+1) \textnormal{ mod } p,k+1))
\]
in $A$. It is easy to see that paths from $(0,0)$ to $(n,t_L,n/4)$ (respectively, $(n,t-t_L,n/4)$) correspond to elements of $\mathcal{L}$ (respectively, $\mathcal{R}$) and that these paths can be enumerated in the claimed time bound with standard (backtracking) methods.

Also, on Line~\ref{lin:binsearcht}, we can enumerate all relevant $Y \in \mathcal{R}$ in $O^*(1)$ time per time per item of $\mathcal{R}$ with binary search using the sorted data structure constructed on Line~\ref{lin:binsearch}.

Hence, Algorithm~\ref{alg:repmet} can be implemented such that it runs in $O^*(2^{0.45n}+|\mathcal{L}|+|\mathcal{R}|+P)$ time, where $P$ is the number of pseudo-solutions $(X,Y) \in \mathcal{L}\times\mathcal{R}$ such that $w(X)= t_L (\mathrm{mod}\ p)$. Since $t_L$ is picked uniformly at random, we have by \textbf{A1} that the expected number of such pseudo-solutions is at most $2^{0.9n} /p =2^{0.45n}$. Thus, under assumption \textbf{A1}, Algorithm~\ref{alg:repmet} runs in $O^*(2^{0.45n})$ time.

\begin{algorithm}
	\caption{Representation Method for \textsc{Subset Sum}}
	\label{alg:repmet}
	\begin{algorithmic}[1]
		\REQUIRE $\mathtt{RepMethod}(w : [n]\rightarrow [W],t)$
		\ENSURE \true with probability at least $1/2$, if a $X \subseteq [n]$ exists with $w(X)=t$; \false otherwise
		\STATE Sample uniformly and independently a prime $p \in \{2^{0.45n},\ldots,2^{0.45n+1} \}$ and $t_L \in \mathbb{Z}_p$
        \STATE Construct
        \[
            \mathcal{L} := \left\{ X \in \binom{[n]}{n/4} : w(X) \equiv_p t_L \right\},  \qquad  \mathcal{R} := \left\{ Y \in \binom{[n]}{n/4} : w(Y) \equiv_p t- t_L \right\}.
        \]\label{lin:clist}
        \STATE Sort $\mathcal{R}$ using as key $w(Y)$ for each $Y \in \mathcal{R}$\label{lin:binsearch}
        \FORALL{$X \in \mathcal{L}$}
        \FORALL{$Y \in \mathcal{R}$ such that $w(X)+w(Y)=t$ }\label{lin:binsearcht}
        \STATE If $X$ and $Y$ are disjoint, \algorithmicreturn \true\label{linov}
        \ENDFOR
        \ENDFOR
        \STATE \algorithmicreturn \false
	\end{algorithmic}
\end{algorithm}

\begin{theorem}
If assumptions \textbf{A1}-\textbf{A3} are satisfied, then Algorithm~\ref{alg:repmet} outputs in $O^*(2^{0.45n})$ time \true with at least constant probability.
\end{theorem}
\begin{proof}[Proof sketch.]
Since the runtime under assumption~\textbf{A1} was already discussed, we only focus on the correctness.
It is clear that if the algorithm outputs 
\true, it found two disjoint sets $X,Y$ with $w(X)+w(Y)=t$ and hence $X \cup Y$ is a solution. By assumption \textbf{A2} and \textbf{A3}, it remains to show that if a set $S$ with $|S|=n/2$ and $w(2^S)=2^{n/2}$ exists, indeed\true is returned with probability at least $1/2$. Note that $w(2^S)=2^{n/2}$ implies that all subsets of $S$ generate different sums with respect to the integers $w$. Hence $w\left(\binom{S}{n/4}\right)=\binom{n/2}{n/4}$.
It follows from the hashing properties of reducing a number modulo a random prime $p$ that with high probability
\[
    \left|\left\{ x\ (\textnormal{mod } p) : x \in w\left(\binom{S}{n/4}\right) \right\}\right| \geq \Omega(2^{0.45 n}).
\]
Conditioned on this event, we have with constant probability that $t_L$ is picked such that there exists $X \in \binom{S}{n/4}$ with $w(X)\equiv_p t_L$, and hence the pair $(X,Y)$ with $Y=S \setminus X$ will be observed to be disjoint at Line~\ref{linov} (unless\true was already returned in an earlier iteration).
\end{proof}

\subsection{Further recent progress}

The question whether the "meet-in-the-middle barrier" formed by the $O^*(2^{n/2})$ runtime of~\cite{DBLP:journals/jacm/HorowitzS74} can be improved has also been studied (and positively answered)  for similar problems~\cite{DBLP:conf/soda/ChenJRS22,DBLP:conf/esa/MuchaNPW19,DBLP:journals/corr/abs-2403-19117}.

Another popular topic of study is the space usage of algorithms. A famous improvement of~\cite{DBLP:journals/jacm/HorowitzS74} is the $O^*(2^{n/2})$ time and $O^*(2^{n/4})$ space algorithm from~\cite{DBLP:journals/siamcomp/SchroeppelS81}, which recently has been improved to a $O^*(2^{n/2})$ time and $O^*(2^{0.249999n})$ space algorithm~
\cite{DBLP:conf/stoc/NederlofW21} and even further to $O^*(2^{n/2})$ time and $O^*(2^{0.246n})$ space in~\cite{BelovaCKM24}.
If one is restricted to only polynomial space, for a long time the best known algorithm was the na\"ive $O^*(2^n)$ algorithm.
In~\cite{DBLP:conf/eurocrypt/BeckerCJ11} the authors introduce the idea to solve random instances with cycle finding algorithms, which was later extended to a $O^*(2^{0.86n})$ time polynomial space algorithm in~\cite{DBLP:journals/siamcomp/BansalGN018} that assumes random read only access to random bits (not stored in memory). This latter assumption can be removed with recent works on pseudo-randomness~\cite{DBLP:conf/soda/ChenJWW22, DBLP:conf/soda/0003Z23}.

\section{Other topics}\label{sec:other}

As mentioned earlier, this survey by no means aims to be an exhaustive survey of all (recent) developments in the field of fine-grained complexity of hard problems. Here we very briefly discuss (in arbitrary order) a few of such notable directions that could have been included in a longer version of this survey.

\paragraph{Parameterized Complexity of \nptime-Complete Problems.} 
While we mentioned parameterized complexity at the start of this survey, it should be stressed that also within parameterized complexity, quite some works address the fine-grained question discussed in this survey.
For example, for many problems parameterized by treewidth or pathwidth researchers found algorithms running in time $O^*(c^{k})$ and proofs that improvements to $O^*((c-\eps)^k)$ time refute Hypothesis~\ref{seth}~\cite{DBLP:journals/talg/LokshtanovMS18}. Other notable well-studied examples are \textsc{$k$-Path} (see~\cite{DBLP:journals/jcss/BjorklundHKK17} for the currently fastest algorithm in undirected graphs) and \textsc{Feedback Vertex Set} (see~\cite{DBLP:journals/talg/LiN22} for the currently fastest algorithm in undirected graphs).
The fine-grained complexity of $\nptime$-hard ``subset'' problems parameterized by solution size was also shown to have direct implications for the fine-grained complexity parameterized by search-space size via a method called "Monotone Local Search" (see e.g.~\cite{DBLP:journals/jacm/FominGLS19}), in a fashion that is somewhat similar to Algorithm~\ref{alg:ls}.

\paragraph{(Conditional) Lower bounds / Reductions.}
An important and wide topic we glanced over in this (optimistically-oriented) survey are (conditionally) lower bounds. That is, some evidence, or a proof under the assumption of a hypothesis, that certain na\"ive algorithms cannot be improved.

The most popular and relevant hypothesis in this direction is Hypothesis~\ref{seth}, but for surprisingly few problems discussed in this survey researchers were able to derive lower bounds as a consequence of~\ref{seth}.
Some notable lower bounds are a number of equivalences to \textsc{Hitting Set} and \textsc{Set Splitting}~\cite{DBLP:journals/talg/CyganDLMNOPSW16}, and a tight lower bound for pseudo-polynomial run times for \textsc{Subset Sum}~\cite{DBLP:journals/talg/AbboudBHS22} and a large body of lower bounds for run times parameterized by treewidth that was initiated in~\cite{DBLP:journals/talg/LokshtanovMS18},
although the latter two may be viewed more as 
fine-grained \emph{parameterized} complexity results.

An outstanding open question is whether a $O^*((2-\eps)^n)$ time algorithm (for some $\eps >0$) for \textsc{Set Cover} would refute Hypothesis~\ref{seth}, or conversely, whether a refutation of Hypothesis~\ref{seth} would imply a $O^*((2-\eps)^n)$ time algorithm for \textsc{Set Cover} (for some $\eps >0$). 
It is also not clear yet how the new results on \textsc{Set Cover}~\cite{DBLP:journals/corr/abs-2310-11926,DBLP:journals/corr/abs-2311-02774} relate to this.

Such connections were given in~\cite{DBLP:journals/talg/CyganDLMNOPSW16} for the \emph{parity} versions of the problems: Roughly speaking, the number of set covers of a set system on $n$ elements can be counted modulo $2$ in $O^*((2-\eps)^n)$ time (for some $\eps >0$) if and only if the number of solutions to an $n$-variate $k$-CNF formula can be counted in $O^*((2-\eps')^n)$ time (for some $\eps' >0$). These connections were made by relating both problems to the task of computing the parity of the number of independent sets in a bipartite graph. Motivated by this, it may also be interesting to study its decision problem:

\defproblemu{\textsc{Constrained Bipartite Independent Set}}
{A bipartite graph $G=(A \cup B, E)$ and integers $t_A,t_B$}
{Is there an independent set $I$ of $G$ such that $|I \cap A| = t_A$ and $|I \cap B| = t_B$?}

Given the above motivation, we believe it is an interesting question to see whether this problem can be solved in $O^*((2-\eps)^{|A|})$ time, for some $\eps >0$.

The lack of strong lower bounds conditioned on Hypothesis~\ref{seth} led some researchers to study the existence of certain relaxed versions of algorithms (such as proof systems~\cite{DBLP:conf/innovations/CarmosinoGIMPS16} and Merlin-Arthur~\cite{DBLP:conf/coco/Williams16} or polynomial formulations~\cite{DBLP:conf/soda/BelovaKMRRS24, DBLP:conf/sosa/KulikovM24, DBLP:conf/soda/BelovaGKMS23}) to give evidence of the ``hardness of showing hardness''.

\paragraph{Branching Algorithms.}
A notable paradigm that has been very well-studied in the realm of fine-grained complexity of \nptime-complete problems is that of branching algorithms. With this paradigm and advanced analyses (such as "Measure and Conquer"~\cite{DBLP:journals/jacm/FominGK09}), researchers were able to achieve the best worst-case run time bounds in terms of the number of vertices of the input graph for, among others, fundamental problems such as \textsc{Independent Set} and \textsc{Dominating Set}. We refer to~\cite[Chapters 2 and 6]{DBLP:series/txtcs/FominK10} for more details.

\paragraph{OPP algorithms.}
A natural, but not so frequently studied model of exponential time algorithm is that of One-sided Probabilistic Polynomial-time (OPP) algorithms. These are algorithms that run in polynomial time and are always correct when they output\false, but are only guaranteed to output \true on YES-instances with inversely exponentially small probability.\footnote{An example of an OPP algorithm for \textsc{$k$-CNF SAT} would be to simply sample a random solution and output whether it satisfies the given formula.}

While this is a natural and interesting model (since it captures most branching algorithms), a number of interesting lower bounds for algorithms captured by this model are presented in the literature. For example, for $n$-input \textsc{Circuit Sat},\footnote{In this problem one is given a Boolean circuit with $n$ inputs and asked whether there exist an assignment of the inputs that make the circuit evaluate to\true.} no polynomial time Monte Carlo algorithm can output\true with probability $(2-\eps)^{-n}$ under Hypothesis~\ref{eth}, as was shown in~\cite{DBLP:conf/stoc/PaturiP10}, and no polynomial time algorithm can output for every $n$-input \textsc{$k$-CNF SAT} formula a satisfying assignment with probability at least $2^{-n^{1-\Omega(1)}}$, unless the polynomial hierarchy collapses~\cite{DBLP:conf/focs/Drucker13}.

\paragraph{Circuit lower bounds.}

A notable application of algorithms improving over na\"ive algorithms was consolidated recently in a research line targeted at proving circuit lower bounds. It was shown that even tiny improvements over na\"ive algorithms for the problem of satisfiability of Boolean circuits implies circuit lower, and that such algorithms can be given by employing a batch evaluation technique based on a fast matrix multiplication or Lemma~\ref{lem:yates}. We refer to the survey~\cite{williams2014algorithms} for more details.

\paragraph{Coarser-Grained Complexity of \nptime-complete problems.}
Even in the coarser-grained complexity regime there have recently been surprising results and the complexity of some fundamental problems remains elusive. For example,
for \textsc{Subgraph Isomorphism}, it was shown in~\cite{DBLP:journals/jacm/CyganFGKMPS17} that the simple $n^{O(n)}$ time algorithm cannot be improved to $n^{o(n)}$ unless Hypothesis~\ref{eth}, where $n$ denotes the number of vertices of the graphs. On the other hand, somewhat surprisingly, it was shown in~
\cite{DBLP:conf/icalp/BodlaenderNZ16} that \textsc{Subgraph Isomorphism} on planar graphs can be solved in $2^{O(n/ \log n)}$, and not in $2^{o(n / \log n)}$ unless the ETH fails. For the \textsc{Many Visits TSP} problem, an old $n^{O(n)}$ time algorithm from~\cite{DBLP:journals/siamcomp/CosmadakisP84} was recently improved to a $2^{O(n)}$ time algorithm in~\cite{DBLP:journals/talg/BergerKMV20}.

A notable open question in this direction is the complexity of the \textsc{Edge Coloring} problem: can it be solved in $2^{o(n^2)}$ time? See also~\cite{DBLP:conf/sosa/KulikovM24}.

\paragraph{Quantum Speed-ups.} A relatively recent topic is to study how much quantum algorithms can speed up exact algorithms for \nptime-complete problems. A straightforward application of Grover search shows that \textsc{$k$-CNF SAT} can be solved in $2^{n/2}$ time, and a quantum analogue of Hypothesis~\ref{seth} was formulated in~\cite{DBLP:conf/stacs/BuhrmanPS21} that posits that this cannot be significantly improved. 
For other problems, finding a quantum speed-ups is not straightforward, but can still be found (see e.g.~\cite{DBLP:conf/soda/AmbainisBIKPV19}).

\section*{Acknowledgements.}
The author thanks the editors for being patient with this manuscript being overdue, Carla Groenland for useful discussions on Section~\ref{sec:setcov}, L\'aszl\'o Kozma for suggestions that improved the write-up of this survey and the reviewers and editor of IPL for their extensive comments on a previous version of Subsection~\ref{subsec:sparse}.

\bibliographystyle{alpha}
\bibliography{refs}
\end{document}